%% file: arxiv.tex
\documentclass[a4paper]{llncs}
\usepackage{fullpage}
\input{preamble}

\begin{document}

\title{Finite-state Strategies in Delay Games (full version)\thanks{Supported by the projects \myquot{TriCS} (ZI 1516/1-1) and (LO 1174/3-1) of the German Research Foundation (DFG). The work presented here was carried out while the first author was a member of the Logic and Theory of Discrete Systems Group at RWTH Aachen University and the second author was a member of the Reactive Systems Group at Saarland University.}}

\author{Sarah Winter\inst{1} \and Martin Zimmermann\inst{2}}
\institute{Formal Methods and Verification, Université libre de Bruxelles, 1050 Bruxelles, Belgium\\
\email{swinter@ulb.ac.be}
\and
University of Liverpool, Liverpool L69 3BX, United Kingdom\\
\email{martin.zimmermann@liverpool.ac.uk}
}

\maketitle
\begin{abstract}
\input{content/abstract}
\end{abstract}

\section{Introduction}
\label{sec-intro}
\input{content/intro}

\section{Preliminaries}
\label{sec-prel}
\input{content/prelims}

\section{What is a Finite-state Strategy in a Delay Game?}
\label{sec-fsindg}
\input{content/finitestateindelaygame}

\section{Computing Finite-state Strategies for Block Games}
\label{sec-construction}
\input{content/construction}

\section{Discussion}
\label{sec-disc}
\input{content/disc}

\section{Succinctly Implementing Finite-state Strategies for Block Games}
\label{sec-implementingstrats}
\input{content/implementingstrats}

\section{Conclusion}
\label{sec-conc}
\input{content/conc}

\bibliographystyle{elsarticle-num}
\bibliography{biblio}

\appendix
\section{Arena-based Games vs.\ Gale-Stewart Games}
\label{sec-appendix}
\input{content/appendix-arenagames}

\end{document}

%% file: preamble.tex
\usepackage{amsmath}
\usepackage{amssymb}
\usepackage{mathtools}
\usepackage{nicefrac}

\usepackage{todonotes}

\usepackage[utf8]{inputenc}
\usepackage[T1]{fontenc}

\usepackage{tikz}
\usetikzlibrary{automata}
\usetikzlibrary{backgrounds}
\usetikzlibrary{calc}

\usepackage{booktabs}

\tikzset{
  initial text=, shorten >= 1pt, auto, 
  every path/.style={->,-stealth, every loop/.style={-stealth}},
  every initial by arrow/.style={-stealth, initial text={},
  every loop/.style={red,-stealth}},
  every state/.style={minimum size=9mm},
}

\makeatletter
\newcommand{\charfusion}[2]{%
  \def\ch@rfusion##1##2{%
    \ooalign{\hfil$##1#1$\hfil\cr\hfil$##2$\hfil\crcr}}%
      \mathop{%
      \vphantom{#1}%
      \mathpalette\ch@rfusion#2}\displaylimits}
\makeatother

\newcommand{\myquot}[1]{``#1''}

\newcommand{\nats}{\mathbb{N}}
\newcommand{\size}[1]{|#1|}

\renewcommand{\epsilon}{\varepsilon}
\renewcommand{\phi}{\varphi}

\newcommand{\set}[1]{\{#1\}}
\newcommand{\pow}[1]{2^{#1}}

\newcommand{\colvec}[1]{\begin{pmatrix}#1\end{pmatrix}}
\newcommand{\bigvec}[2][]{{#1\begin{pmatrix}#2\end{pmatrix}}}

\newcommand{\aut}{\mathfrak{A}}
\newcommand{\autb}{\mathfrak{B}}
\newcommand{\strataut}{\mathfrak{T}}

\newcommand{\acc}{\mathrm{Acc}}
\newcommand{\col}{\Omega}
\newcommand{\infi}[0]{\mathrm{Inf}}
\newcommand{\curlyF}[0]{\mathcal{F}}

\newcommand{\parity}{\mathrm{prty}}
\newcommand{\muller}{\mathrm{mllr}}
\newcommand{\initmark}{I}

\newcommand{\inp}[1]{\mathit{in}(#1)}
\newcommand{\outp}[1]{\mathit{out}(#1)}

\newcommand{\stratautDelta}{\Delta}
\newcommand{\stratautLambda}{\Lambda}

\newcommand{\fDelta}{\Delta}
\newcommand{\fLambda}{\Lambda}

\newcommand{\delaygame}[1]{\Gamma\!_{f}(#1)}

\newcommand{\delaygameabbr}{\Gamma\!_{f}}

\newcommand{\blockgame}[2]{\Gamma^{#1}(#2)}

\newcommand{\SigmaI}{\Sigma_I}
\newcommand{\SigmaO}{\Sigma_O}
\newcommand{\strat}{\tau}
\newcommand{\stratO}{\tau_O}
\newcommand{\stratI}{\tau_I}
\newcommand{\p}{P}

\newcommand{\arena}{\mathcal{A}}

\newcommand{\game}{\Gamma}
\newcommand{\win}{\mathrm{Win}}

\newcommand{\arenagame}{\mathcal{G}}

\newcommand{\bigo}{\mathcal{O}}

\newcommand{\update}{\mathrm{upd}}

\newcommand{\block}[1]{\overline{#1}}

\newcommand{\equivclass}[1]{[#1]_\equiv}
\newcommand{\R}{{R}}

\newcommand{\mon}{\mathfrak{M}}
\newcommand{\Leq}{L_=}
\newcommand{\rep}{\mathrm{rep}}

%% file: content/abstract.tex
What is a finite-state strategy in a delay game? We answer this surprisingly non-trivial question by presenting a very general framework that allows to remove delay: finite-state strategies exist for all winning conditions where the resulting delay-free game admits a finite-state strategy. The framework is applicable to games whose winning condition is recognized by an automaton with an acceptance condition that satisfies a certain aggregation property.

 Our framework also yields upper bounds on the complexity of determining the winner of such delay games and upper bounds on the necessary lookahead to win the game. In particular, we cover all previous results of that kind as special cases of our uniform approach. 

%% file: content/intro.tex
What is a finite-state strategy in a delay game? The answer to this question is surprisingly non-trivial due to the nature of delay games in which one player is granted a lookahead on her opponent's moves. This puts her into an advantage when it comes to winning games, i.e., there are games that can only be won with lookahead, but not without. A simple example is a game where one has to predict the third move of the opponent with one's first move. This is impossible when moving in alternation, but possible if one has access to the opponent's first three moves before making the first move. More intriguingly, lookahead also allows to improve the quality of winning strategies in games with quantitative winning conditions, i.e., there is a tradeoff between quality and amount of lookahead~\cite{Zimmermann17}. More practically, when modeling reactive synthesis as a two-player game, the addition of delay allows us to model delay inherent to sensing and actuating in the physical world as well as the delay caused by the transmission of data~\cite{ChenFLMZ18}. Thus, delay games capture aspects of real-life synthesis problems that cannot easily be expressed in the classical, i.e., delay-free, framework. 

However, managing (and, if necessary, storing) the additional information gained by the lookahead can be a burden. Consider another game where one just has to copy the opponent's moves. This is obviously possible with or without lookahead (assuming the opponent moves first). In particular, without lookahead one just has to remember the last move of the opponent and copy it. However, when granted lookahead, one has to store the last moves of the opponent in a queue to implement the copying properly. This example shows that lookahead is not necessarily advantageous when it comes to minimizing the memory requirements of a strategy.

In this work, we are concerned with Gale-Stewart games~\cite{GaleStewart53}, abstract games without an underlying arena.\footnote{The models of Gale-Stewart games and arena-based games are interreducible, but delay games are naturally presented as a generalization of Gale-Stewart games. This is the reason we prefer this model here.} In such a game, both players produce an infinite sequence of letters and the winner is determined by the combination of these sequences. If it is in the winning condition, a set of such combinations, then the second player wins, otherwise the first one wins. In a classical Gale-Stewart game, both players move in alternation while in a delay game, the second player skips moves to obtain a lookahead on the opponent's moves. Which moves are skipped is part of the rules of the game and known to both players.

Delay games have recently received a considerable amount of attention after being introduced by Hosch and Landweber~\cite{HoschLandweber72} only three years after the seminal Büchi-Landweber theorem~\cite{BuechiLandweber69}. Büchi and Landweber had shown how to solve infinite two-player games with $\omega$-regular winning conditions. Forty years later, delay games were revisited by Holtmann, Kaiser, and Thomas~\cite{HoltmannKaiserThomas12} and the first comprehensive study was initiated, which settled many basic problems like the exact complexity of solving $\omega$-regular delay games and the amount of  lookahead necessary to win such games~\cite{KleinZimmermann16}. Furthermore, Martin's seminal Borel determinacy theorem~\cite{Martin75} for Gale-Stewart games has been lifted to delay games~\cite{KleinZimmermann15} and winning conditions beyond the $\omega$-regular ones have been investigated~\cite{FridmanLoedingZimmermann11,KleinZimmermann16a,Zimmermann16,Zimmermann17}.  

Finally, the decision version of the uniformization problem is to decide whether a given relation has a uniformization, that is, whether there exists a function with a prescribed property that is contained in the relation and has the same domain.
This problem for relations over infinite words and continuous functions boils down to solving delay games: a relation~$L \subseteq (\SigmaI \times \SigmaO)^\omega$ is uniformized by a continuous function (in the Cantor topology) if, and only if, the delaying player wins the delay game with winning condition~$L$. We refer to~\cite{HoltmannKaiserThomas12} for details.

What makes finite-state strategies in infinite games particularly useful and desirable is that a general strategy is an infinite object, as it maps finite play prefixes to next moves. On the other hand, a finite-state strategy is implemented by a transducer, an automaton with output, and therefore finitely represented: the automaton reads a play prefix and outputs the next move to be taken. Thus, the transducer computes a finite abstraction of the play's history using its state space as memory and determines the next move based on the current memory state.

 In Gale-Stewart games, finite-state strategies suffice for all $\omega$-regular games~\cite{BuechiLandweber69} and even for deterministic $\omega$-contextfree games, if one allows pushdown transducers~\cite{Walukiewicz01}. For Gale-Stewart games (and arena-based games), the notion is well-established and one of the most basic questions about a class of winning conditions is that about the existence and size of winning strategies for such games.

While foundational questions for delay games have been answered and many results have been lifted from Gale-Stewart games to those with delay, the issue of computing tractable and implementable strategies has not been addressed before. However, this problem is of great importance, as the existence and computability of finite-state strategies is a major reason for the successful application of infinite games to diverse problems like reactive synthesis, model-checking of fixed-point logics, and automata theory. 

In previous work, restricted classes of strategies for delay games have been considered~\cite{KleinZimmermann15}. However, those restrictions are concerned with the amount of information about the lookahead's evolution a strategy has access to, and do not restrict the size of the strategies: In general, they are still infinite objects. On the other hand, it is known that bounded lookahead suffices for many winning conditions of importance, e.g., the $\omega$-regular ones~\cite{KleinZimmermann16}, those recognized by parity and Streett automata with costs~\cite{Zimmermann17}, and those definable in (parameterized) linear temporal logics~\cite{KleinZimmermann16a}. Furthermore, for all those winning conditions, the winner of a delay game can be determined effectively. In fact, all these proofs rely on the same basic construction that was already present in the work of Holtmann, Kaiser, and Thomas, i.e., a reduction to a Gale-Stewart game using equivalence relations that capture behavior of the automaton recognizing the winning condition. These reductions and the fact that finite-state strategies suffice for the games obtained in the reductions imply that (some kind of) finite-state strategies exist for such games.
 
Indeed, in his master's thesis, Salzmann recently introduced the first notion of finite-state strategies in delay games and, using these reductions, presented an algorithm computing them for several types of acceptance conditions, e.g., parity conditions and related $\omega$-regular ones~\cite{Salzmann15}. However, the exact nature of finite-state strategies in delay games is not as canonical as for Gale-Stewart games. We discuss this issue in-depth in Sections~\ref{sec-fsindg} and \ref{sec-disc} by proposing two notions of finite-state strategies, a delay-oblivious one which yields large strategies in the size of the lookahead, and a delay-aware one that follows naturally from the reductions to Gale-Stewart games mentioned earlier. In particular, the number of states of the delay-aware strategies is independent of the size of the lookahead, but often larger in the size of the automaton recognizing the winning condition. However, this is offset by the fact that strategies of the second type are simpler to compute than the delay-oblivious ones and have overall fewer states, if the lookahead is large.  In comparison to Salzmann's notion, where strategies syntactically depend on a given automaton representing the winning condition, our strategies are independent of the representation of the winning condition and therefore more general. Also, our framework is more abstract and therefore applicable to a wider range of acceptance conditions (e.g., qualitative ones) and yields in general smaller strategies, but there are of course some similarities, which we discuss in detail.

To present these notions, we first introduce some definitions in Section~\ref{sec-prel}, e.g., delay games and finite-state strategies for Gale-Stewart games. After introducing the two notions of finite-state strategies for delay games in Section~\ref{sec-fsindg}, we show how to compute such strategies in Section~\ref{sec-construction}. To this end, we present a generic account of the reduction from delay games to Gale-Stewart games which subsumes, to the best of our knowledge, all decidability results presented in the literature. Furthermore, we show how to obtain the desired strategies from our construction. Then, in Section~\ref{sec-disc}, we compare the different definitions of finite-state strategies for delay games proposed here and discuss their advantages and disadvantages. Also, we compare our approach to that of Salzmann. In Section~\ref{sec-implementingstrats}, we discuss how to implement finite-state strategies in delay games even more succinctly. We conclude by mentioning some directions for further research in Section~\ref{sec-conc}. 

This work is an extended and revised version of a paper presented at GandALF 2017~\cite{Zimmermann17c}.

\paragraph*{Related Work}

As mentioned earlier, the existence of finite-state strategies is the technical core of many applications of infinite games, e.g., in reactive synthesis one synthesizes a correct-by-construction system from a given specification by casting the problem as an infinite game between a player representing the system and one representing the antagonistic environment. It is a winning strategy for the system player that yields the desired implementation, which is finite if the winning strategy is finite-state. Similarly, Gurevich and Harrington's game-based proof of Rabin's decidability theorem for monadic second-order logic over infinite binary trees~\cite{Rabin1969} relies on the existence of finite-state strategies.\footnote{The proof is actually based on positional strategies, a further restriction of finite-state strategies for arena-based games, because they are simpler to handle. Nevertheless, the same proof also works for finite-state strategies.}

These facts explain the need for studying the existence and properties of finite-state strategies in infinite games~\cite{Khoussainov03,Rabinovich09,LeRouxPauly16,Thomas94}. In particular, the seminal work by Dziembowski, Jurdzi\'{n}ski, and Walukiewicz~\cite{DziembowskiJW97} addressed the problem of determining upper and lower bounds on the size of finite-state winning strategies in games with Muller winning conditions. Nowadays, one of the most basic questions about a given winning condition is that about such upper and lower bounds. For most conditions in the literature, tight bounds are known, see, e.g.,~\cite{ChatterjeeHenzingerHorn11,Horn05,WallmeierHuettenThomas03}. But there are also surprising exceptions to that rule, e.g., generalized reachability games~\cite{FijalkowH13}. More recently, Colcombet, Fijalkow, and Horn presented a very general technique that yields tight upper and lower bounds on memory requirements in safety games, which even hold for games in  infinite arenas, provided their degree is finite~\cite{ColcombetFH14}.

%% file: content/prelims.tex
We denote the non-negative integers by~$\nats$. 
An alphabet~$\Sigma$ is a non-empty finite set. 
The set of finite words over $\Sigma$ is denoted by $\Sigma^*$ and the set of infinite words by $\Sigma^\omega$.
Given a finite or infinite word $\alpha$, we denote by $\alpha(i)$ the $i$th letter of $\alpha$, starting with $0$, i.e., $\alpha = \alpha(0)\alpha(1)\alpha(2)\cdots$.
Given two $\omega$-words~$\alpha \in (\Sigma_0)^\omega$ and $\beta \in (\Sigma_1)^\omega$, we define ${ \alpha \choose \beta} = {\alpha(0) \choose \beta(0)} {\alpha(1) \choose \beta(1)} {\alpha(2) \choose \beta(2)} \cdots \in (\Sigma_0 \times \Sigma_1)^\omega$. Similarly, we define ${ x \choose y }$ for finite words~$x$ and $y$ with $\size{x} = \size{y}$.

\subsection{$\omega$-automata}

A (deterministic and complete) $\omega$-automaton is a tuple~$\aut = (Q, \Sigma, q_\initmark, \delta, \acc)$ where $Q$ is a finite set of states, $\Sigma$ is an alphabet, $q_\initmark \in Q$ is the initial state, $\delta \colon Q \times \Sigma \rightarrow Q$ is the transition function, and $\acc \subseteq \delta^\omega$ is the set of accepting runs (here, and whenever convenient, we treat $\delta$ as a relation~$\delta \subseteq Q \times \Sigma \times Q$).
A finite run~$\pi$ of $\aut$ is a sequence
\[
\pi = (q_0, a_0, q_1)(q_1, a_1, q_2) \cdots (q_{i-1}, a_{i-1}, q_i) \in \delta^+.\]
As usual, we say that $\pi$ starts in $q_0$, ends in $q_i$, and processes~$a_0\cdots a_{i-1} \in \Sigma^+$. Infinite runs on infinite words are defined analogously. If we speak of \emph{the} run of $\aut$ on $\alpha \in \Sigma^\omega$, then we mean the unique run of $\aut$ starting in $q_\initmark$ processing $\alpha$. The language~$L(\aut) \subseteq \Sigma^\omega$ of $\aut$ contains all those $\omega$-words whose run of $\aut$ is accepting. The size of $\aut$ is defined as $\size{\aut}=\size{Q}$.

This definition is very broad, which allows us to formulate our theorems as general as possible. In examples, we consider safety, reachability, parity, and Muller automata whose sets of accepting runs are finitely represented: An $\omega$-automaton~$\aut = ( Q, \Sigma, q_\initmark, \delta, \acc )$ is a safety automaton, if there is a set~$F  \subseteq Q$ of accepting states such that 
\[\acc = \set{(q_0, a_0, q_1) (q_1, a_1, q_2) (q_2, a_2, q_3) \cdots \in \delta^\omega \mid q_i \in F \text{ for every } i}.\]
Moreover, an $\omega$-automaton~$\aut = ( Q, \Sigma, q_\initmark, \delta, \acc )$ is a reachability automaton, if there is a set~$F  \subseteq Q$ of accepting states such that 
\[\acc = \set{(q_0, a_0, q_1) (q_1, a_1, q_2) (q_2, a_2, q_3) \cdots \in \delta^\omega \mid q_i \in F \text{ for some } i}.\]
	Furthermore, $\aut$ is a parity automaton, if 
\[\acc = \set{(q_0, a_0, q_1) (q_1, a_1, q_2) (q_2, a_2, q_3) \cdots \in \delta^\omega \mid \text{$\limsup\nolimits_{i \rightarrow \infty} \col(q_i)  $ is even}}\] for some coloring~$\col \colon Q \rightarrow \nats$. To simplify our notation, define $\col(q,a,q') = \col(q)$. Finally, $\aut$ is a Muller automaton, if there is a family~$\curlyF \subseteq \pow{Q}$ of sets of states such that $\acc = \set{ \rho \in \delta^\omega \mid \infi(\rho) \in \curlyF}$, where $\infi(\rho)$ is the set of states visited infinitely often by $\rho$. 

\subsection{Delay Games}

A delay function is a mapping~$f \colon \nats \rightarrow \nats\setminus \set{0}$, which is said to be constant if $f(i) =1$ for all $i>0$. A delay game~$\delaygame{L}$ consists of a delay function~$f$ and a winning condition~$L \subseteq (\SigmaI \times \SigmaO)^\omega$ for some alphabets~$\SigmaI$ and $\SigmaO$. Such a game is played in rounds~$i = 0,1,2, \ldots$ as follows: in round~$i$, first Player~$I$ picks a word~$x_i \in \SigmaI^{f(i)}$, then Player~$O$ picks a letter~$y_i \in \SigmaO$. Player~$O$ wins a play~$(x_0, y_0)(x_1, y_1)(x_2, y_2) \cdots $ if the outcome~${ x_0 x_1 x_2 \cdots \choose y_0 y_1 y_2 \cdots }$ is in $L$; otherwise, Player~$I$ wins.

A strategy for Player~$I$ in $\delaygame{L}$ is a mapping $\stratI \colon \SigmaO^* \rightarrow \SigmaI^*$ satisfying $\size{\stratI(w)} = f(\size{w})$ while a strategy for Player~$O$ is a mapping~$\stratO \colon \SigmaI^+ \rightarrow \SigmaO$. A play~$(x_0, y_0)(x_1, y_1)(x_2, y_2) \cdots $ is consistent with $\stratI$ if $x_i = \stratI(y_0 \cdots y_{i-1})$ for all $i$, and it is consistent with $\stratO$ if $y_i = \stratO(x_0 \cdots x_i)$ for all $i$. A strategy for Player~$\p \in \set{I,O}$ is winning, if every play that is consistent with the strategy is won by Player~$\p$.

An important special case are delay-free games, i.e., those with respect to the delay function~$f$ mapping every $i$ to $1$. In this case, we drop the subscript~$f$ and write $\game(L)$ for the game with winning condition~$L$. Such games are typically called Gale-Stewart games~\cite{GaleStewart53}.

\subsection{Finite-state Strategies in Gale-Stewart Games}
\label{subsec-finitestate4galestewart}
A strategy for Player~$O$ in a Gale-Stewart game is still a mapping~$\stratO \colon \SigmaI^+ \rightarrow \SigmaO$. Such a strategy is said to be finite-state, if there is a deterministic finite transducer~$\strataut$ that implements $\stratO$ in the following sense: $\strataut$ is a tuple~$(Q, \SigmaI, q_\initmark, \delta, \SigmaO, \lambda)$ where $Q$ is a finite set of states, $\SigmaI$ is the input alphabet, $q_\initmark \in Q$ is the initial state, $\delta \colon Q \times \SigmaI \rightarrow Q$ is the deterministic transition function, $\SigmaO$ is the output alphabet, and $\lambda \colon Q \rightarrow \SigmaO$ is the output function. Let $\delta^*(q,x)$ denote the unique state that is reached by $\strataut$ when processing $x \in \SigmaI^*$ from $q \in Q$. Then, the strategy~$\strat_\strataut$ implemented by $\strataut$ is defined as $\strat_\strataut(x) = \lambda(\delta^*(q_\initmark,x))$. We say that a strategy is finite-state, if it is implementable by some transducer. Slightly abusively, we identify finite-state strategies with transducers implementing them and talk about finite-state strategies with some number of states. Thus, we focus on the \emph{state complexity} (e.g., the number of memory states necessary to implement a strategy) and ignore the other components of a transducer (which are anyway of polynomial size in $\size{Q}$, if we assume $\SigmaI$ and $\SigmaO$ to be fixed).

%% file: content/finitestateindelaygame.tex
Before we answer this question, we first ask what properties a finite-state strategy should have, i.e., what makes finite-state strategies in Gale-Stewart games useful and desirable? A strategy~$\stratO \colon \SigmaI^+ \rightarrow \SigmaO$ is in general an infinite object and does not necessarily have a finite representation. Furthermore, to execute such a strategy, one needs to store the whole sequence of moves made by Player~$I$ thus far: Unbounded memory is needed to execute it. 

On the other hand, a finite-state strategy is finitely described by an automaton~$\strataut$ implementing it.  To execute it, one only needs to store a single state of $\strataut$ and access to the transition function~$\delta$ and the output function~$\lambda$ of $\strataut$. Assume the current state is $q$ at the beginning of some round~$i$ (initialized with $q_\initmark$ before round~$0$). Then, Player~$I$ makes his move by picking some~$a \in \SigmaI$, which is processed by updating the memory state to $q' = \delta(q, a)$. Then, $\strataut$ prescribes picking $\lambda(q') \in \SigmaO$ and round~$i$ is completed.

Thus, there are two aspects that make finite-state strategies desirable: (1) the next move depends only on a finite amount of information about the history of the play, i.e., a state of the automaton, which is (2) easily updated. In particular, the necessary \emph{machinery} of the strategy is encoded in the transition function and the output function. 

Further, there is a generic framework to compute such strategies by reducing them to arena-based games.\footnote{\ref{sec-appendix} gives a short introduction to arena-based games.} As an example, consider a game~$\game(L(\aut))$ where $\aut$ is a parity automaton with set~$Q$ of states and transition function~$\delta$. We describe the construction of an arena-based parity game contested between Player~$I$ and Player~$O$ whose solution allows us to compute the desired strategies (formal details are presented in the appendix). The positions of Player~$I$ are states of $\aut$ while those of Player~$O$ are pairs~$(q,a)$ where $q \in Q$ and where $a$ is an input letter. From a vertex~$q$ Player~$I$ can move to every state~$(q,a)$ for $a \in \SigmaI$, from which Player~$O$ can move to every vertex~$\delta(q,{a \choose b})$ for $b \in \SigmaO$. Finally, Player~$O$ wins a play, if the run of $\aut$ constructed during the play is accepting. It is easy to see that the resulting game is a parity game with~$ \size{Q}\cdot(\size{\SigmaI}+1)$ vertices, and has the same winner as $\game(L(\aut))$. The winner of the arena-based game has a positional\footnote{A strategy in an arena-based game is positional, if only depends on the last vertex of the play's history, not on the full history. A formal definition can be found in the appendix.} winning strategy~\cite{EmersonJutla91,Mostowski91}, which can be computed in quasipolynomial time~\cite{CJKLS16,FJSSW17,JL17,DBLP:conf/lics/Lehtinen18}. Such a positional winning strategy can easily be turned into a finite-state winning strategy with $\size{Q} \cdot \size{\SigmaI}$ states for Player~$O$ in the game~$\game(L(\aut))$, which is implemented by an automaton with state set~$Q \times \SigmaI$. 
This reduction can be generalized to arbitrary classes of Gale-Stewart games whose winning conditions are recognized by an $\omega$-automaton with set~$Q$ of states: if Player~$O$ has a finite-state strategy with $n'$ states in the arena-based game obtained by the construction described above, then Player~$O$ has a finite-state winning strategy with $\size{Q} \cdot \size{\SigmaI}\cdot n'$ states for the original Gale-Stewart game. Such a strategy is obtained by solving an arena-based game with $\size{Q}\cdot(\size{\SigmaI}+1)$ vertices. Again, see the appendix for technical details.

\subsection{Delay-oblivious Finite-state Strategies in Delay Games}

So, what is a finite-state strategy in a delay game? In the following, we discuss this question for the case of delay games with respect to constant delay functions, which is the most important case. In particular, constant lookahead suffices for all $\omega$-regular winning conditions~\cite{KleinZimmermann16}, i.e, Player~$O$ wins with respect to an arbitrary delay function if, and only if, she wins with respect to a constant one. Similarly, constant lookahead suffices for many quantitative conditions like (parameterized) temporal logics~\cite{KleinZimmermann16a} and parity conditions with costs~\cite{Zimmermann17}. For winning conditions given by parity automata, there is an exponential upper bound on the necessary constant lookahead. On the other hand, there are exponential lower bounds already for winning conditions specified by deterministic automata with reachability or safety acceptance~\cite{KleinZimmermann16} (which are subsumed by parity acceptance).

Technically, a strategy for Player~$O$ in a delay game is still a mapping~$\stratO \colon \SigmaI^+ \rightarrow \SigmaO$. Hence, the definition of finite-state strategies for Gale-Stewart games (see Subsection~\ref{subsec-finitestate4galestewart}) is also applicable to delay games. With reasons that become apparent in the example succeeding the definition, we call such strategies \emph{delay-oblivious}. 

As a (cautionary) example, consider a delay game with winning condition~$\Leq = \set{ {\alpha \choose \alpha} \mid \alpha \in \set{0,1}^\omega }$, i.e., Player~$O$ just has to copy Player~$I$'s moves, which she can do with respect to every delay function: Player~$O$ wins $\delaygame{\Leq}$ for every $f$. However, a finite-state strategy has to remember the whole lookahead, i.e., those moves that Player~$I$ is ahead of Player~$O$, in order to copy his moves. Thus, an automaton implementing a winning strategy for Player~$O$ in $\delaygame{\Leq}$ needs at least~$\size{\set{0,1}}^d $ states, if $f$ is a constant delay function with $f(0) = d$. Thus, the memory requirements grow with the size of the lookahead granted to Player~$O$, i.e., lookahead is a burden, not an advantage. She even needs unbounded memory in the case of unbounded lookahead.

An advantage of this delay-oblivious definition is that finite-state strategies can be obtained by a trivial extension of the reduction presented for Gale-Stewart games above: now, states of Player~$I$ are from~$Q \times \SigmaI^{d-1}$ and those of Player~$O$ are from $Q \times \SigmaI^{d}$. Player~$I$ can move from $(q,w)$ to $(q, wa)$ for $a \in \SigmaI$ while Player~$O$ can move from $(q, aw)$ to $(\delta(q, {a \choose b}),w)$ for $b \in \SigmaO$. Intuitively, a state now additionally stores a queue of length~$d-1$, which contains the lookahead granted to Player~$O$. Coming back to the parity example, this approach yields a finite-state strategy with $\size{Q}\cdot\size{\SigmaI}^{d}$ states. To obtain such a strategy, one has to solve a parity game with $\size{Q}\cdot(\size{\SigmaI}+1)\cdot\size{\SigmaI}^{d-1}$ vertices, which is of doubly-exponential size in $\size{\aut}$, if $d$ is close to the (tight) exponential upper bound. This can be done in doubly-exponential time, as it still has the same number of colors as the automaton~$\aut$. Again, this reduction can be generalized to arbitrary classes of delay games with constant delay whose winning conditions are recognized by an $\omega$-automaton with set~$Q$ of states: if Player~$O$ has a finite-state strategy with $n'$ states in the arena-based game obtained by the construction, then Player~$O$ has a finite-state winning strategy with $\size{Q} \cdot \size{\SigmaI}^d\cdot n'$ states for the delay game with constant lookahead of size~$d$. In general, $d$ factors exponentially into $n'$, as $n'$ is the memory size required to win a game with $\bigo(\size{\SigmaI}^d)$ vertices. Also, to obtain the strategy for the delay game, one has to solve an arena-based game  with $\size{Q}\cdot(\size{\SigmaI}+1)\cdot \size{\SigmaI}^{d-1}$~vertices. Again, see the appendix for technical details.

\subsection{Block Games}

We show that one can do better by decoupling the history tracking and the handling of the lookahead, i.e., by using \emph{delay-aware} finite-state strategies. In the delay-oblivious definition, we hardcode a queue into the arena-based game, which results in a blowup of the arena and therefore also in a blowup in the solution complexity and in the number of memory states for the arena-based game, which is turned into one for the delay game. To overcome this, we introduce a slight variation of delay games with respect to constant delay functions, so-called block games\footnote{Holtmann, Kaiser, and Thomas already introduced a notion of block game in connection to delay games~\cite{HoltmannKaiserThomas12}. However, their notion differs from ours in several aspects. Most importantly, in their definition, Player~$I$ determines the length of the blocks (within some bounds specified by $f$) while our block length is fixed and part of the rules of the game.}, present a notion of finite-state strategy in block games, and show how to transfer strategies between delay games and block games. Then, we show how to solve block games and how to obtain finite-state strategies for them.

 The motivation for introducing block games is to eliminate the queue containing the letters Player~$I$ is ahead of Player~$O$, which is cumbersome to maintain, and causes the blowup in the case of games with winning condition~$\Leq$. Instead, in a block game, both players pick blocks of letters of a fixed length with Player~$I$ being one block ahead to account for the delay, i.e., Player~$I$ has to pick two blocks in round~$0$ and then one in every round, as does Player~$O$ in every round. This variant of delay games lies implicitly or explicitly at the foundation of all arguments establishing upper bounds on the necessary lookahead and at the foundations of all algorithms solving delay games~\cite{HoltmannKaiserThomas12,KleinZimmermann16,KleinZimmermann16a,Zimmermann16,Zimmermann17}. Furthermore, we show how to transform a (winning) strategy for a delay game into a (winning) strategy for a block game and vice versa, i.e., Player~$O$ wins the delay game if, and only if, she wins the corresponding block game.\footnote{Due to their importance and prevalence for solving delay games, one could even argue that the notion of block games is more suitable to model delay in infinite games.}

Formally, the block game~$\blockgame{d}{L}$, where $d \in \nats\setminus\set{0}$ is the block length and where $L \subseteq (\SigmaI\times\SigmaO)^\omega$ is the winning condition, is played in rounds as follows: in round~$0$, Player~$I$ picks two blocks~$\block{a_0},\block{a_1} \in \SigmaI^d$, then Player~$O$ picks a block~$\block{b_0}\in\SigmaO^d$. In round~$i>0$, Player~$I$ picks a block~$\block{a_{i+1}} \in \SigmaI^d$, then Player~$O$ picks a block~$\block{b_i}\in\SigmaO^d$. Player~$O$ wins the resulting play~$\block{a_0} \block{a_1} \block{b_0} \block{a_2} \block{b_1} \cdots$, if the outcome~${\block{a_0}\block{a_1}\block{a_2} \cdots  \choose \block{b_0}\block{b_1}\block{b_2} \cdots}$ is in $L$. 

A strategy for Player~$I$ in $\game$ is a map~$\stratI \colon (\SigmaO^{d})^* \rightarrow (\SigmaI^{d})^2 \cup \SigmaI^{d}$ such that $\stratI(\epsilon) \in (\SigmaI^{d})^2$ and $\stratI(\block{b_0} \cdots \block{b_i}) \in \SigmaI^{d}$ for $i\ge 0$. A strategy for Player~$O$ is a map~$\stratO \colon (\SigmaI^d)^* \rightarrow \SigmaO^d$. A play~$\block{a_0} \block{a_1} \block{b_0} \block{a_2} \block{b_1} \cdots$ is consistent with $\stratI$, if $(\block{a_0}, \block{a_1}) = \stratI(\epsilon)$ and $\block{a_i} = \stratI(\block{b_0} \cdots \block{b_{i-2}})$ for every $i \ge 2$; it is consistent with $\stratO$ if $\block{b_i} = \stratO(\block{a_0} \cdots \block{a_{i+1}})$ for every $i \ge 0$. Winning strategies and winning a block game are defined as for delay games.

The next lemma relates delay games with constant lookahead and block games: for a given winning condition, Player~$O$ wins a delay game with winning condition~$L$ (with respect to some delay function) if, and only if, she wins a block game with winning condition~$L$ (for some block size). 

\begin{lemma}
\label{lemma-delaygamesvsblockgames}
Let $L \subseteq (\SigmaI \times \SigmaO)^\omega$.
\begin{enumerate}
	\item\label{lemma-delaygamesvsblockgames-delay2block}
	 If Player~$O$ wins $\delaygame{L}$ for some constant delay function~$f$, then she also wins $\blockgame{f(0)}{L}$.
	\item\label{lemma-delaygamesvsblockgames-block2delay}
	 If Player~$O$ wins $\blockgame{d}{L}$, then she also wins $\delaygame{L}$ for the constant delay function~$f$ with $f(0) = 2d$. 
\end{enumerate}
\end{lemma}

\begin{proof}
\ref{lemma-delaygamesvsblockgames-delay2block}) 
Let $\stratO \colon \SigmaI^+ \rightarrow \SigmaO$ be a winning strategy for Player~$O$ in $\delaygame{L}$ and fix $d = f(0)$. Now, define $\stratO' \colon (\SigmaI^d)^* \rightarrow (\SigmaO)^d$ for Player~$O$ in $\blockgame{d}{L}$ via $\stratO'(\block{a_0} \cdots \block{a_i} \block{a_{i+1}}) = \beta(0) \cdots \beta(d-1)$ with
$\beta(j) = \stratO(\block{a_0} \cdots \block{a_{i}} \alpha(0) \alpha(1) \cdots \alpha(j-1))$
for $\block{a_{i+1}} = \alpha(0) \alpha(1) \cdots \alpha(d-1)$.

A straightforward induction shows that for every play consistent with $\stratO'$ there is a play consistent with $\stratO$ that has the same outcome. Thus, as $\stratO$ is a winning strategy, so is $\stratO'$.

\ref{lemma-delaygamesvsblockgames-block2delay}) Now, let $\stratO' \colon (\SigmaI^d)^* \rightarrow (\SigmaO)^d$ be a winning strategy for Player~$O$ in $\blockgame{d}{L}$. We define $\stratO \colon \SigmaI^+ \rightarrow \SigmaO$ for Player~$O$ in $\delaygame{L}$. To this end, let $x \in \SigmaI^+$. By the choice of $f$, we obtain $\size{x} \ge f(0) = 2d$. Thus, we can decompose $x$ into $x = \block{a_0} \cdots \block{a_{i}} x'$ such that $i \ge 1$, each $\block{a_{i'}}$ is a block over $\SigmaI$ and $\size{x'} < d$. 
Now, let $\stratO'(\block{a_0} \cdots \block{a_i}) = \beta(0) \cdots \beta(d-1)$. Then, we define $\stratO(x) = \beta(\size{x'})$.

Again, a straightforward induction shows that for every play consistent with $\stratO$ there is a play consistent with $\stratO'$ that has the same outcome. Thus, $\stratO$ is a winning strategy.
\end{proof}

\subsection{Delay-aware Finite-state Strategies in Block Games}
\label{subsec-delayaware}

After having proved the equivalence of block games and delay games w.r.t.\ constant delay, we now define \emph{delay-aware} finite-state strategies for block games. 
Fix a block game~$\blockgame{d}{L}$ with $L \subseteq (\SigmaI \times \SigmaO)^\omega$. A finite-state strategy for Player~$O$ in $\blockgame{d}{L}$ is implemented by a transducer~$\strataut = (Q, \SigmaI, q_\initmark, \delta, \SigmaO, \lambda)$ where $Q$, $\SigmaI$, and $q_\initmark$ are defined as in Subsection~\ref{subsec-finitestate4galestewart}. However, the transition function~$\delta\colon Q \times \SigmaI^d \rightarrow Q$ processes full input blocks and the output function~$\lambda \colon Q \times \SigmaI^d \times \SigmaI^d \rightarrow \SigmaO^d$ maps a state and a pair of input blocks to an output block. The strategy~$\strat_\strataut$ implemented by $\strataut$ is defined as $\strat_\strataut(\block{a_0} \cdots \block{a_i}) = \lambda(\delta^*(q_\initmark,\block{a_0} \cdots \block{a_{i-2}}), \block{a_{i-1}}, \block{a_i} )$ for $i \ge 1$. Here, $\delta^*(q, \block{a_0} \cdots \block{a_{i-2}})$ is the state reached by $\strataut$ when processing $\block{a_0} \cdots \block{a_{i-2}}$ from $q$.

\begin{example} Fix some $d >0$. 
Player~$O$ has a trivial delay-aware finite-state winning strategy for $\blockgame{d}{\Leq}$ which is implemented by the transducer~$(Q, \SigmaI, q_\initmark, \delta, \SigmaO, \lambda)$ where $Q = \set{q_\initmark}$, $\delta(q_\initmark, \block{a}) = q_\initmark$ for every $q \in Q$ and every $\block{a} \in \SigmaI^d$, and where $\lambda(q, \block{a_0}, \block{a_1}) = \block{a_0}$  for every $q \in Q$ and every $\block{a_0},\block{a_1} \in \SigmaI^d$.
\end{example}

Again, we identify delay-aware strategies with transducers implementing them and are interested in the number of states of the transducer. This captures the amount of information about a play's history that is differentiated in order to implement the strategy. Note that this ignores the  representation of the transition and the output function. These are no longer \myquot{small} (in $\size{Q}$), as it is the case for transducers implementing strategies for Gale-Stewart games. When focusing on executing such strategies, these factors become relevant, but for our current purposes they are not: 
We have  decoupled the history tracking from the lookahead-handling. The former is implemented by the automaton as usual while the latter is taken care of by the output function. In particular, the size of the automaton is (a-priori) independent of the block size. In Section~\ref{sec-implementingstrats}, we revisit the issue of representing the transition and the output function succinctly, thereby addressing the issue of implementability. 

In the next section, we present a very general approach to computing finite-state strategies for block games whose winning conditions are specified by automata with acceptance conditions that satisfy a certain aggregation property. For example, for block games with winning conditions given by deterministic parity automata, we obtain a strategy implemented by a transducer with exponentially many states, which can be obtained by solving a parity game of exponential size. In both aspects, this is an exponential improvement over the delay-oblivious variant for classical delay games.

To conclude the introduction of block games, we strengthen Lemma~\ref{lemma-delaygamesvsblockgames} to transfer finite-state strategies between delay games and block games.

\begin{lemma}\label{lemma-delaygamesvsblockgames-fs}
Let $L \subseteq (\SigmaI \times \SigmaO)^\omega$.
\begin{enumerate}
	\item\label{lemma-delaygamesvsblockgames-fs_delay2block}
	 If Player~$O$ has a delay-oblivious finite-state winning strategy for $\delaygame{L}$ with $n$ states for some constant delay function~$f$, then she also has a delay-aware finite-state winning strategy for $\blockgame{f(0)}{L}$ with $n$ states.
	\item\label{lemma-delaygamesvsblockgames-fs_block2delay}
	 If Player~$O$ has a delay-aware finite-state winning strategy for $\blockgame{d}{L}$ with $n$ states, then she also has a delay-oblivious finite-state winning strategy for $\delaygame{L}$  with $n \cdot \size{\SigmaI}^{2d} $ states for the constant delay function~$f$ with $f(0) = 2d$. 
\end{enumerate}
\end{lemma}

\begin{proof}
It is straightforward to achieve the strategy transformations described in the proof of Lemma~\ref{lemma-delaygamesvsblockgames} by transforming transducers that implement finite-state strategies. 
\end{proof}

The blowup in the direction from block games to delay games is in general unavoidable, as finite-state winning strategies for the game~$\delaygame{\Leq}$ need at least $2^d$ states to store the lookahead while winning strategies for the block game need only one state, independently of the block size.

%% file: content/construction.tex
The aim of this section is twofold. Our main aim is to compute finite-state strategies for block games (and, by extension, for delay games with constant lookahead). We do so by presenting a general framework for analyzing delay games with winning conditions specified by $\omega$-automata whose acceptance conditions satisfy a certain aggregation property. The technical core is a reduction to a Gale-Stewart game, i.e., we remove the delay from the game. This framework yields upper bounds on the necessary (constant) lookahead to win a given game, but also allows to determine the winner and a finite-state winning strategy, if the resulting Gale-Stewart game can be effectively solved. 

Slightly more formally, let $\aut$ be the automaton recognizing the winning condition of the block game. Then, the winning condition of the Gale-Stewart game constructed in the reduction is recognized by an automaton~$\autb$ that can be derived from $\aut$. In particular, the acceptance condition of $\autb$ simulates the acceptance condition of $\aut$. Many types of acceptance conditions are preserved by the simulation, e.g., starting with a parity automaton~$\aut$, we end up with a parity automaton~$\autb$. Thus, the resulting Gale-Stewart game can be effectively solved.

Our second aim is to present a framework as general as possible to obtain upper bounds on the necessary lookahead and on the solution complexity for a wide range of winning conditions. In fact, our framework is a generalization and abstraction of techniques first developed for the case of $\omega$-regular winning conditions~\cite{KleinZimmermann16}, which were later generalized to other winning conditions~\cite{KleinZimmermann16a,Zimmermann16,Zimmermann17}. Here, we cover all these results in a uniform way.

Let us begin by giving some intuition for the construction. The winning condition of the game is recognized by an automaton~$\aut$. Thus, as usual, the exact input can be abstracted away, only the induced behavior in $\aut$ is relevant. Such a behavior is characterized by the state transformations induced by processing the input and by the effect on the acceptance condition triggered by processing it. For many acceptance conditions, this effect can be aggregated, e.g., for parity conditions, one can decompose runs into non-empty pieces and then only consider the maximal colors of the pieces. For many quantitative winning conditions, one additionally needs bounds on the lengths of these pieces (cf.~\cite{Zimmermann16,Zimmermann17}).

We first introduce aggregations and give some examples in Subsection~\ref{subsec_aggregations} before we present the reduction to Gale-Stewart games using aggregations in Subsection~\ref{subsec_reduction}

\subsection{Aggregations}
\label{subsec_aggregations}
We begin by introducing two types of aggregations of varying strength. Fix an $\omega$-automaton~$\aut = (Q, \Sigma, q_\initmark, \delta, \acc )$ and let $s \colon \delta^+ \rightarrow M$ for some finite set~$M$. Given a decomposition~$(\pi_i)_{i\in\nats}$ of a run $\pi_0 \pi_1 \pi_2 \cdots$ into non-empty pieces~$\pi_i\in \delta^+$ we define $s((\pi_i)_{i\in\nats}) = s(\pi_0)s(\pi_1)s(\pi_2)\cdots \in M^\omega$.

\begin{itemize}
	
	
	\item We say that $s$ is a strong aggregation (function) for $\aut$, if for all decompositions~$(\pi_i)_{i\in\nats}$ and $(\pi_i')_{i\in\nats}$ of any runs~$\rho = \pi_0 \pi_1 \pi_2 \cdots $ and $\rho' = \pi_0' \pi_1' \pi_2' \cdots $ with $\sup_i \size{\pi_i'} < \infty$ and $s((\pi_i)_{i\in\nats}) = s((\pi_i')_{i\in\nats})$: $\rho \in \acc \Rightarrow \rho'\in \acc$.

	\item We say that $s$ is a weak  aggregation (function) for $\aut$, if for all decompositions~$(\pi_i)_{i\in\nats}$ and $(\pi_i')_{i\in\nats}$  of any runs~$\rho = \pi_0 \pi_1 \pi_2 \cdots $ and $\rho' = \pi_0' \pi_1' \pi_2' \cdots $ with $\sup_i \size{\pi_i} < \infty$, $\sup_i \size{\pi_i'} < \infty$, and $s((\pi_i)_{i\in\nats}) = s((\pi_i')_{i\in\nats})$: $\rho \in \acc \Rightarrow \rho'\in \acc$.

\end{itemize}

Thus, in a strong aggregation, only the pieces~$\pi_i'$ of $\rho'$ are of bounded length while in a weak aggregation both the pieces~$\pi_i$ of $\rho$ and the pieces~$\pi_i'$ of $\rho'$ are of bounded length.

\begin{example}\hfill
\label{example-aggregation}

\begin{itemize}
	
	\item The function~$s_\parity \colon \delta^+ \rightarrow \col(Q)$ defined as $s_\parity(t_0 \cdots t_i) = \max_{0\le j \le i}\col(t_j)$ is a strong aggregation for a parity automaton~$(Q, \Sigma, q_\initmark, \delta, \acc)$ with coloring~$\col$.
	
	\item The function~$s_\muller \colon \delta^+ \rightarrow \pow{Q}$ defined as $s_\muller((q_0, a_0, q_1) \cdots (q_{n}, a_n, q_{n+1})) = \set{q_0, q_1, \ldots, q_{n}}$ is a strong aggregation for a Muller automaton~$(Q, \Sigma, q_\initmark, \delta, \acc)$. 
	
	\item The exponential time algorithm for delay games with winning conditions given by parity automata with costs, a quantitative generalization of parity automata, is based on a strong aggregation~\cite{Zimmermann17}.
	
	\item The algorithm for delay games with winning conditions given by max automata~\cite{Bojanczyk11}, another quantitative automaton model, is based on a weak aggregation~\cite{Zimmermann16}.
	
\end{itemize}
\end{example}

Due to symmetry, we can replace the implication~$\rho \in \acc \Rightarrow \rho' \in \acc$ by an equivalence in the definition of a weak aggregation. Also, the notions are trivially hierarchical, i.e., every strong aggregation is also a weak one. 

Let us briefly comment on the difference between strong and weak aggregations using the examples of parity automata with costs and max-automata: the acceptance condition of the former automata is a boundedness condition on some counters while the acceptance condition of the latter is a boolean combination of boundedness and unboundedness conditions on some counters. The aggregations for these acceptance conditions capture whether a piece of a run induces an increment of a counter or not, but abstract away the actual number of increments if it is non-zero. 
Now, consider the parity condition with costs, which requires to bound the counters. Assume the counters in some run~$\pi_0 \pi_1 \pi_2 \cdots$ are bounded and that we have pieces~$\pi_i'$ of bounded length having the same aggregation. Then, the increments in some piece~$\pi_i'$ have at least one corresponding increment in $\pi_i$. Thus, if a counter in $\pi_0' \pi_1' \pi_2' \cdots$ is unbounded, then it is also unbounded in $\pi_0 \pi_1 \pi_2 \cdots$, which yields a contradiction. Hence, the implication~$\pi_0 \pi_1 \pi_2 \cdots \in \acc \Rightarrow \pi_0' \pi_1' \pi_2' \cdots\in \acc$ holds. For details, see~\cite{Zimmermann17}.
On the other hand, to preserve boundedness and unboundedness properties, one needs to bound the length of the $\pi_i'$ and the length of the $\pi_i$. Hence, there is only a weak aggregation for max-automata. Again, see~\cite{Zimmermann16} for details.

Given a weak  aggregation~$s$ for $\aut$ with acceptance condition~$\acc$, let 
\[
s(\acc) = \set{
s((\pi_i)_{i\in\nats}) \mid \pi_0 \pi_1 \pi_2 \cdots \in \acc \text{ is an accepting run of }\aut \text{ with } \sup\nolimits_i \size{\pi_i} < \infty}
.\]

Next, we consider aggregations that are trackable by automata. A monitor for an automaton~$\aut$ with transition function~$\delta$ is a tuple~$\mon = (M,\bot, \update)$ where $M$ is a finite set of memory elements, $\bot \notin M$ is the empty memory element, and $\update \colon M_\bot \times \delta \rightarrow M$ is an update function, where we use $M_\bot = M \cup \set{\bot}$. Note that the empty memory element~$\bot$ is only used to initialize the memory, it is not in the image of $\update$. 
We say that $\mon$ computes the function~$s_\mon \colon \delta^+ \rightarrow M$ defined by $s_\mon(t) = \update(\bot, t)$ and $s_\mon(\pi\cdot t) = \update(s_\mon(\pi), t)$ for $\pi \in \delta^+$ and $t \in \delta$. 

\begin{example}
\label{example-monitor}
Recall  Example~\ref{example-aggregation}. The strong aggregation~$s_\parity$ for a parity automaton is computed by the monitor~$(\col(Q), \bot, (c,t) \mapsto \max\set{c,\col(t)})$, where $\bot < c$ for every $c \in \col(Q)$.
\end{example}

Finally, we take the product of $\aut$ and the monitor~$\mon$ for $\aut$, which simulates $\aut$ and simultaneously aggregates the acceptance condition. Formally, we define the product as $\aut \times \mon = (Q \times M_\bot, (q_\initmark, \bot), \Sigma, \delta', \emptyset)$ where $\delta'((q, m), a) = (q', \update(m, (q, a, q'))) $ for $q' = \delta(q, a)$. Note that $\aut\times\mon$ has an empty set of accepting runs, as these are irrelevant to us.

\subsection{Removing Delay via Aggregations}
\label{subsec_reduction}
Consider a play prefix in a delay game~$\delaygame{L(\aut)}$: Player~$I$ has produced a sequence~$\alpha(0) \cdots \alpha(i)$ of letters while Player~$O$ has produced $\beta(0) \cdots \beta(i')$ with, in general, $i'<i$. Now, she has to determine $\beta(i'+1)$. The automaton~$\aut \times \mon$ can process the joint sequence~${\alpha(0) \cdots \alpha(i') \choose \beta(0)\cdots \beta(i')}  $, but not the sequence~$\alpha(i'+1) \cdots \alpha(i)$, as Player~$O$ has not yet picked the letters~$\beta(i'+1) \cdots \beta(i)$. However, one can determine which states are reachable by some completion~${\alpha(i'+1) \cdots \alpha(i) \choose \beta(i'+1) \cdots \beta(i)}$ by  projecting away $\SigmaO$ from $\aut\times\mon$.

Thus, from now on assume $\Sigma = \SigmaI \times \SigmaO$ and define $\delta_P \colon 2^{Q\times M_\bot} \times \SigmaI \rightarrow \pow{Q \times M}$ ($P$ for power set) via \label{page_deltap}
\[
\delta_P (S, a) = \left\{\left.\delta'\left( (q,m), {a \choose b} \right) \right| (q,m) \in S \text{ and } b \in \SigmaO\right\}.\]
Intuitively, $\delta_P$ is obtained as follows: take $\aut\times\mon$, project away $\SigmaO$, and apply the power set construction (while discarding the anyway empty acceptance condition). Then, $\delta_P$ is the transition function of the resulting deterministic automaton.
As usual, we extend $\delta_P$ to $\delta_P^+ \colon 2^{Q \times M_\bot} \times \SigmaI^+ \rightarrow 2^{Q\times M}$ via $\delta_P^+(S,a) = \delta_P(S,a)$ and $\delta_P^+(S, wa) = \delta_P( \delta_P^+(S, w), a)$. 

Given states~$q$ and $q'$ of $\aut$, a memory state~$m$, and a word~$w \in \SigmaI^+$, we call a word~$w' \in \SigmaO^{\size{w}}$ a $(q,q',m)$-completion of $w$, if the run~$\pi$ of $\aut$ processing ${w \choose w'}$ starting from $q$ ends in $q'$ and satisfies $s_\mon(\pi) = m$.

\begin{remark}
\label{remark-powersetcharac}
The following are equivalent for $q \in Q$ and $w \in \SigmaI^+$:
\begin{enumerate}
	\item $(q',m') \in \delta_P^+(\set{(q, \bot)},w)$. 
	\item There is a $(q,q',m')$-completion of $w$.
\end{enumerate}
\end{remark}

We use this property to define an equivalence relation formalizing the idea that words having the same behavior in $\aut \times \mon$ do not need to be distinguished. To this end, to every $w \in \SigmaI^+$ we assign the transition summary~$r_w \colon Q \rightarrow \pow{Q \times M}$ defined via $
r_w(q) = \delta_P^+(\set{(q, \bot)}, w)$.
Having the same transition summary is a finite equivalence relation~$\equiv$ over $\SigmaI^+$ whose index is bounded by $2^{\size{Q}^2\size{M}}$. For an $\equiv$-class~$S = \equivclass{w}$ define $r_S = r_w$, which is independent of representatives. Let $\R$ be the set of infinite $\equiv$-classes. 

Now, we define a Gale-Stewart game in which Player~$I$ determines an infinite sequence of equivalence classes from $\R$. By picking representatives, this induces a word~$\alpha \in \SigmaI^\omega$. Player~$O$ picks states~$(q_i,m_i)$ such that the $m_i$ aggregate a run of $\aut$ on some completion~${\alpha \choose \beta}$ of $\alpha$. Player~$O$ wins if the $m_i$ imply that the run of $\aut$ on ${\alpha \choose \beta}$ is accepting. To account for the delay, Player~$I$ is always one move ahead, which is achieved by adding a dummy move for Player~$O$ in round~$0$.

Formally, in round~$0$, Player~$I$ picks an $\equiv$-class~$S_0 \in \R$ and Player~$O$ has to pick $(q_0, m_0) = (q_\initmark, \bot)$. In round~$i>0$, first Player~$I$ picks an $\equiv$-class~$S_i \in \R$, then Player~$O$ picks a state~$(q_i, m_i) \in r_{S_{i-1}}(q_{i-1})$ of the product automaton. Player~$O$ wins the resulting play~$S_0 (q_0, m_0) S_1 (q_1, m_1) S_2 (q_2, m_2) \cdots$ if $m_1 m_2 m_3 \cdots \in s_\mon(\acc)$ (note that $m_0$ is ignored). The notions of (finite-state and winning) strategies are inherited from Gale-Stewart games, as this game is indeed such a game~$\game(L(\autb))$ for some automaton~$\autb$ of size~$\size{\R} \cdot \size{Q}\cdot\size{M} $ which can be derived from $\aut$ and $\mon$ as follows:

Let $\aut = (Q, \SigmaI\times \SigmaO, q_\initmark, \delta, \acc)$ and $\mon = (M, \bot, \update)$ be given.
We define $\autb = (\R \times Q \times M_\bot, \R \times (Q \times M), (S_\initmark, q_\initmark, m_\initmark), \delta', \acc')$ for some arbitrary~$S_\initmark \in \R$, some arbitrary~$m_\initmark \in M$, $\delta'((S,q,m),{S' \choose (q',m')}) = (S',q',m')$, and 
$(S_0,q_0,m_0) (S_1,q_1,m_1) (S_2,q_2,m_2) \cdots \in \acc'$ if, and only if, 
\begin{itemize}
	\item $(q_0, m_0) = (q_\initmark,\bot)$,
	\item $(q_i,m_i) \in r_{S_{i-1}}(q_{i-1})$ for all $i >0$, and
	\item $m_1 m_2 m_3 \cdots \in s_\mon(\acc)$.
\end{itemize}
It is straightforward to prove that $\autb$ has the desired properties.

Note that, due to our very general definition of acceptance conditions, we are able to express the local consistency requirement~\myquot{$(q_i,m_i) \in r_{S_{i-1}}(q_{i-1})$} using the acceptance condition. For less general acceptance modes, e.g., parity, one has to check this property using the state space of the automaton, which leads to a polynomial blowup, as one has to store each~$S_{i-1}$ for one transition.

\begin{theorem}\label{thm-main}
Let $\aut$ be an $\omega$-automaton and let $\mon$ be a monitor for $\aut$ such that $s_\mon$ is a strong aggregation for $\aut$, let $\autb$ be constructed as above, and define $d = 2^{\size{Q}^2\cdot\size{M_\bot}}$.

\begin{enumerate}
	
	\item\label{thm-main-delay2delayfree} If Player~$O$ wins $\delaygame{L(\aut)}$ for some delay function~$f$, then she also wins $\game(L(\autb))$.
	
	\item\label{thm-main-delayfree2block} If Player~$O$ wins $\game(L(\autb))$, then she also wins the block game~$\blockgame{d}{L(\aut)}$. Moreover, if she has a finite-state winning strategy for $\game(L(\autb))$ with $n$ states, then she has a delay-aware finite-state winning strategy for $\blockgame{d}{L(\aut)}$ with $n $ states.
	
\end{enumerate}
\end{theorem}

Before we prove these results, we need to establish a closure property of the sets~$s(\acc)$ in case $s$ is a strong aggregation for an $\omega$-automaton with acceptance condition~$\acc$. Recall that we defined 
\[
s(\acc) = \set{
s((\pi_i)_{i\in\nats}) \mid \pi_0 \pi_1 \pi_2 \cdots \in \acc \text{ is an accepting run of }\aut \text{ with } \sup\nolimits_i \size{\pi_i} < \infty},
\]
i.e., $s(\acc)$ only contains the aggregations of decompositions into pieces of bounded length. However, if $s$ is strong, then this restriction is not essential: the $\pi_i$ in the following lemma are not required to be of bounded length.

\begin{lemma}
\label{lemma-saccclosure}
Let $s$ be a strong aggregation for an $\omega$-automaton~$\aut$ with acceptance condition~$\acc$ and let $\pi_0 \pi_1 \pi_2 \cdots \in \acc$. Then, $s((\pi_i)_{i\in\nats}) \in s(\acc)$.
\end{lemma}

\begin{proof}
The $s$-profile of a finite run~$\pi$ is the tuple~$(q, m, q')$ where $q$ is the state $\pi$ starts in, $m = s(\pi)$, and $q'$ is the state $\pi$ ends in. Having the same $s$-profile is an equivalence relation over finite runs of finite index. For each equivalence class~$S$ of this relation, let $\rep(S)$ be an arbitrary, but fixed, element of $S$. For notational convenience, define $\rep(\pi) = \rep(S)$ for the unique equivalence class~$S$ with $\pi \in S$.

Now, consider the sequence~$\rep(\pi_0) \rep(\pi_1) \rep(\pi_2) \cdots$. By construction, it is also a run of $\aut$ and we have $s((\pi_i)_{i\in\nats}) = s((\rep(\pi_i))_{i\in\nats})$. As the $\rep(\pi_i)$ are of bounded length (after all, there are only finitely many representatives), $s$ being a strong aggregation yields $\rep(\pi_0) \rep(\pi_1) \rep(\pi_2) \cdots \in \acc$. Hence,
$s((\pi_i)_{i\in\nats}) = s((\rep(\pi_i))_{i\in\nats}) \in s(\acc)$.
\end{proof}

Also, we need some basic properties of equivalence classes~$S \in R$. They follow from the fact that every equivalence class~$S$, which is a language of finite words, is recognized by a deterministic finite automaton of size~$d$ (as defined above) obtained using the state set~$(\pow{Q \times M_\bot})^{{Q}}$ to simulate $\delta_P^+$ starting from the states of the form~$\set{(q, \bot)}$.

\begin{remark}\label{remark-equivbounds}
\hfill
\begin{enumerate}
	\item\label{remark-equivbounds-finitelymanynotinR} $\equivclass{w} \in \R$ for every $w$ of length at least $d$. 
	\item\label{remark-equivbounds-density} Let $S \in\R$. For every $n$, $S$ contains a word~$w$ of length~$n \le \size{w} \le n+d$.
\end{enumerate}
\end{remark}

Now, we are ready to prove Theorem~\ref{thm-main}.

\begin{proof} The argument is a further generalization of similar constructions for parity automata (with or without costs) and max-automata (cp.~\cite{KleinZimmermann16,Zimmermann16,Zimmermann17}).

\ref{thm-main-delay2delayfree}) Let $\stratO^f \colon \SigmaI^+ \rightarrow \SigmaO$ be a winning strategy for Player~$O$ in $\delaygame{L(\aut)}$ for some fixed~$f$. For the sake of readability, we denote $\delaygame{L(\aut)}$ by $\delaygameabbr$ and $\game(L(\autb))$ by $\Gamma$. We describe how to simulate a play in $\Gamma$ by a play in $\delaygameabbr$ to transform $\stratO^f$ into a winning strategy~$\stratO$ for Player~$O$ in $\Gamma$. 

To this end, let Player~$I$ pick $S_0 \in \R$ in $\Gamma$, which has to be answered by Player~$O$ by picking~$(q_0, m_0) = (q_\initmark, \bot)$. Thus, we define $\stratO(S_0) = (q_0, m_0)$. Next, Player~$I$ picks some $S_1 \in \R$.

To simulate this, pick some $x_0 \in S_0$ satisfying $\size{x_0} \ge f(0)$, which exists due to $S_0$ being infinite by virtue of being in $\R$. Similarly, we pick some $x_1 \in S_1$ satisfying $\size{x_0x_1} \ge \sum_{j=0}^{\size{x_0}-1}f(j)$, which again exists due to $S_1$ being infinite. 

Now, assume Player~$I$ starts a play by picking the letters of the prefix of $x_0 x_1$ of length~$\sum_{j=0}^{\size{x_0}-1}f(j)$ during the first $\size{x_0}$ rounds of $\delaygameabbr$. By the choice of $\size{x_1}$, $x_0 x_1$ is long enough to do so. Let $y_0$ be the answer of Player~$O$ according to $\stratO$ during these $\size{x_0}$ rounds, i.e., $\size{y_0} = \size{x_0}$. 

Thus, we are in the following situation for $i = 1$:
\begin{itemize}
	\item In $\Gamma$, the players have produced the play prefix~$S_0 (q_0, m_0) \cdots S_{i-1} (q_{i-1}, m_{i-1}) S_i$.
	\item In $\delaygameabbr$, Player~$I$ has picked a prefix of $x_0 \cdots x_i$ while Player~$O$ has picked $y_0 \cdots y_{i-1}$ according to $\stratO^f$ such that $\size{y_0 \cdots y_{i-1}} = \size{x_0 \cdots x_{i-1}}$. Furthermore, we have $\equivclass{x_j} = S_j$ for every $j \le i$.
\end{itemize}

Now, let $i >0$ be arbitrary. Let $q_i$ be the state reached by $\aut$ when processing ${x_{i-1} \choose y_{i-1}}$ when starting in $q_{i-1}$, let $\pi_{i-1}$ be the corresponding run, and define $m_i = s_\mon(\pi_{i-1})$. Then, by definition of $r_{S_i}$ and by Remark~\ref{remark-powersetcharac}, $\stratO(S_0 \cdots S_i) = (q_i, m_i)$ is a legal move in $\Gamma$, which is answered by Player~$I$ picking some $S_{i+1} \in \R$. Again, we pick some $x_{i+1} \in S_{i+1}$ such that  $\size{x_0 \cdots x_{i+1}} \ge \sum_{j=0}^{\size{x_0 \cdots x_i}-1}f(j)$ and consider the play prefix of $\delaygameabbr$ where Player~$I$ starts by picking the letters of the prefix of $x_0 \cdots x_{i+1}$ of length~$ \sum_{j=0}^{\size{x_0 \cdots x_i}-1}f(j)$ during the first $\size{x_0 \cdots x_i}$ rounds, which is a continuation of the previously defined one. Player~$O$ answers the letters of $x_i$ by some $y_i$ of the same length. Thus, we are in the situation above for $i+1$, which concludes the inductive definition of~$\stratO$. 

To conclude, we show that $\stratO$ is indeed winning for Player~$O$ in $\Gamma$. So, let $w = S_0 (q_0, m_0) S_1 (q_1, m_1) S_2 (q_2, m_2) \cdots$ be a play consistent with $\stratO$ and let $w^f = { x_0\choose y_0}{ x_1\choose y_1}{x_2 \choose y_2}\cdots$ be the outcome of the simulated play of $\delaygameabbr$ as described above. By construction, $w^f$ is in $L(\aut)$, as the simulated play is consistent with the winning strategy~$\stratO^f$. 

Let $\pi_{i}$ be defined as above, i.e., $\pi_0 \pi_1 \pi_2 \cdots$ is the run of $\aut$ on $w^f$ and therefore accepting. By construction, we have $s_\mon(\pi_{i}) = m_{i+1}$. Applying Lemma~\ref{lemma-saccclosure} yields $m_1 m_2 m_3 \cdots = s_\mon((\pi_{i})_{i \in \nats}) \in s_\mon(\acc)$. Thus, $w$ is indeed winning for Player~$O$.

\ref{thm-main-delayfree2block}) Let $\stratO$ be a winning strategy for Player~$O$ in $\game(L(\autb))$. Again, for the sake of readability, we denote $\game(L(\autb))$ by $\Gamma$ and $\blockgame{d}{L(\aut)}$ by $\Gamma^d$. As before, we simulate a play in $\Gamma^d$ by a play in $\Gamma$ to transform $\stratO$ into a winning strategy~$\stratO^d$ for Player~$O$ in $\Gamma^d$. In the following proof, all blocks~$\block{a_i}$ are in $\SigmaI^d$ and all $\block{b_i}$ are in $\SigmaO^d$. 

Thus, let Player~$I$ pick $\block{a_0}$ and $\block{a_1}$ during the first round in $\Gamma^d$ and define~$S_0 = \equivclass{\block{a_0}}$, $(q_0, m_0) = \stratO(S_0)$, $S_1 = \equivclass{\block{a_1}}$, and $(q_1,m_1) = \stratO(S_0S_1)$. 

Now, we are in the following situation for $i = 1$. 
\begin{itemize}
	\item In $\Gamma^d$ Player~$I$ has picked $\block{a_0} \cdots \block{a_i}$ and Player~$O$ has picked~$\block{b_0} \cdots \block{b_{i-2}}$ (which is empty for $i=1$).
	\item In $\Gamma$, we have constructed the play prefix~$S_0 (q_0, m_0) \cdots S_{i-1} (q_{i-1}, m_{i-1}) S_i (q_i, m_i)$ that is consistent with $\stratO$ and satisfies $S_j = \equivclass{\block{a_j}}$ for every $j \le i$. 
\end{itemize}

Now, let $i > 0$ be arbitrary. By definition of $\Gamma$, we have $(q_i, m_i) \in r_{\block{a_{i-1}}}(q_{i-1})$. Thus, by definition of $r_{\block{a_{i-1}}}$ and Remark~\ref{remark-powersetcharac} there is a $\block{b_{i-1}}$ such that the run~$\pi_{i-1}$ of $\aut$ processing~${\block{a_{i-1}} \choose \block{b_{i-1}}}$ from $q_{i-1}$ ends in $q_i$ and satisfies $s_\mon(\pi_{i-1}) = m_i$. We define $\stratO^d(\block{a_0} \cdots \block{a_i}) = \block{b_{i-1}}$. This move is answered by Player~$I$ picking some block~$a_{i+1}$, which again induces~$S_{i+1} = \equivclass{\block{a_{i+1}}}$. Applying $\stratO$ yields~$(q_{i+1}, m_{i+1}) = \stratO(S_0 \cdots S_{i+1})$. Thus, we are in the situation described above for $i+1$, which completes the inductive definition of $\stratO^d$.

Note that if $\stratO$ is implemented by a transducer~$\strataut$ with $n$ states, then $\stratO^d$ can easily be implemented by an automaton with $n $ states, which is obtained from $\strataut$ as follows: we use the same set of states so that processing~$\block{a_0} \cdots \block{a_{i-2}}$ leads to the state reached when processing~$\equivclass{\block{a_{0}}} \cdots \equivclass{\block{a_{i-2}}}$, call it~$q$. Now, assume we have two additional blocks~$\block{a_{i-1}}$ and $\block{a_{i}}$ and have to compute the block~$\block{b_{i-1}} = \stratO^d(\block{a_{0}} \cdots \block{a_{i}})$ as defined above. This block only depends on the state~$q$ of the automaton implementing the strategy, on the states~$q_{i-1}$ and $q_i$ of $\aut$, on $m_i$, and on $\block{a_{i-1}}$. All this information can be computed from $q$ and the moves~$\equivclass{\block{a_{i-1}}}$ and $\equivclass{\block{a_{i}}}$ of Player~$I$ in the simulating play. 

It remains to show that $\stratO^d$ is indeed a winning strategy for Player~$O$ in $\Gamma^d$. To this end, let $w^d = \block{a_0}\block{a_1} \block{b_0} \block{a_2} \block{b_1} \cdots$ a play that is consistent with $\stratO^d$. Furthermore, let $S_0 (q_0, m_0) S_1 (q_1, m_1) S_2 (q_2, m_2) \cdots$ be the simulated play in $\Gamma$ constructed as described above, which is consistent with $\stratO$. Therefore, it is winning for Player~$O$, i.e., $m_1 m_2 m_3 \cdots \in s_\mon(\acc)$.

Let the finite runs~$\pi_i$ be defined as above, i.e., $\pi_0 \pi_1 \pi_2 \cdots $ is the run of $\aut$ on $w^d$ and the part~$\pi_i$ processes ${\block{a_i} \choose \block{b_i}}$. Thus, the length of each $\pi_i$ is equal to $d$. Furthermore, we have $s_\mon(\pi_i) = m_{i+1}$ for every $i$. From $s_\mon((\pi_i)_{i\in \nats}) = m_1 m_2 m_3 \cdots \in s_\mon(\acc)$ and $s_\mon$ being a weak aggregation (as it is strong), we conclude that $\pi_0 \pi_1 \pi_2 \cdots $ is accepting. Hence, $w^d \in L(\aut)$, i.e., Player~$O$ wins the play. 
\end{proof}

By applying both implications and Item~\ref{lemma-delaygamesvsblockgames-block2delay} of Lemma~\ref{lemma-delaygamesvsblockgames}, we obtain upper bounds on the complexity of determining for a given~$\aut$ whether Player~$O$ wins $\delaygame{L(\aut)}$ for some $f$ and on the constant lookahead necessary to do so.

\begin{corollary}
Let $\aut$, $\mon$, $\autb$, and $d$ be as in Theorem~\ref{thm-main}. Then, the following are equivalent:
\begin{enumerate}
	\item Player~$O$ wins $\delaygame{L(\aut)}$ for some delay function~$f$.
	\item Player~$O$ wins $\delaygame{L(\aut)}$ for the constant delay function~$f$ with $f(0) = 2d$.
	\item Player~$O$ wins $\game(L(\autb))$.
\end{enumerate}
\end{corollary}

Thus, determining whether, given $\aut$, Player~$O$ wins $\delaygame{L(\aut)}$ for some $f$ is achieved by determining the winner of the Gale-Stewart game~$\game(L(\autb))$ and, independently, we obtain an exponential upper bound on the necessary constant lookahead (in $\size{Q} \cdot \size{M}$).

\begin{example}
Continuing our example for the parity acceptance condition, we obtain the exponential upper bound~$2^{\size{Q}^2\cdot \size{\col(Q)}+2}$ on the constant lookahead necessary to win the delay game and an exponential-time algorithm for determining the winner, as $\autb$ has exponentially many states, but the same number of colors as $\aut$. Both upper bounds are tight~\cite{KleinZimmermann16}.
\end{example}

In case there is no strong aggregation for $\aut$, but only a weak  one, one can show that finite-state strategies exist, if Player~$O$ wins with respect to some constant delay function at all.

\begin{theorem}\label{thm-main2}

Let $\aut$ be an $\omega$-automaton and let $\mon$ be a monitor for $\aut$ such that $s_\mon$ is a weak aggregation for $\aut$, let $\autb$ be constructed as above, and define $d = 2^{\size{Q}^2\cdot\size{M_\bot}}$.
\begin{enumerate}
	
	\item\label{thm-main2-delay2delayfree} If Player~$O$ wins $\delaygame{L(\aut)}$ for some constant delay function~$f$, then she also wins $\game(L(\autb))$.
	
	\item\label{thm-main2-delayfree2block} If Player~$O$ wins $\game(L(\autb))$, then she also wins the block game~$\blockgame{d}{L(\aut)}$. Moreover, if she has a finite-state winning strategy for $\game(L(\autb))$  with $n$ states, then she has a delay-aware finite-state winning strategy for $\blockgame{d}{L(\aut)}$ with $n $ states.
		
\end{enumerate}

\end{theorem}

\begin{proof}
The second implication is the same as the second one in Theorem~\ref{thm-main}, in whose proof we only required $s$ to be a weak  aggregation, which is the setting here. Hence, we only have to consider the first implication.

To this end, we construct a strategy~$\stratO$ for Player~$O$ in $\game(L(\autb))$ from a winning strategy~$\stratO^f$ for Player~$O$ in $\delaygame{L(\aut)}$ as described in the proof of Item~\ref{thm-main-delay2delayfree} of Theorem~\ref{thm-main}. The only difference is that here we can ensure that the length of the $x_i$ is bounded, as $f$ is constant. This allows us to replace the invocation of Lemma~\ref{lemma-saccclosure} and directly apply the definition of $s_\mon(\acc)$ to show that the plays consistent with $\stratO$ are winning for Player~$O$. 
\end{proof}

Again, we obtain upper bounds on the solution complexity (here, with respect to constant delay functions) and on the necessary constant lookahead.

\begin{corollary}
Let $\aut$, $\mon$, $\autb$, and $d$ be as in Theorem~\ref{thm-main2}. Then, the following are equivalent:
\begin{enumerate}
	\item Player~$O$ wins $\delaygame{L(\aut)}$ for some constant delay function~$f$.
	\item Player~$O$ wins $\delaygame{L(\aut)}$ for the constant delay function~$f$ with $f(0) = 2d$.
	\item Player~$O$ wins $\game(L(\autb))$.
\end{enumerate}
\end{corollary}

%% file: content/disc.tex
Let us compare the two approaches presented in the previous section with three use cases: delay games whose winning conditions are given by deterministic parity automata, by deterministic Muller automata, and by LTL formulas. All formalisms only define $\omega$-regular languages, but vary in their succinctness. 

The following facts about arena-based games will be useful for the comparison:
\begin{itemize}
	\item The 
	winner of an arena-based parity game 
 has a positional winning strategy~\cite{EmersonJutla91,Mostowski91}, i.e., a finite-state strategy with a single state.
	\item 
	The winner of an arena-based Muller game has a finite-state strategy with $n!$ states~\cite{McNaughton93}, where $n$ is the number of vertices of the arena.
	\item 
	The winner of an arena-based LTL game has a finite-state strategy with $2^{2^{\bigo(\size{\phi})}}$ states~\cite{PnueliRosner89a}, where $\phi$ is the formula specifying the winning condition.
\end{itemize}
Also, we need the following bounds on the necessary lookahead in delay games: 
\begin{itemize}
	\item In delay games whose winning conditions are given by deterministic parity automata, exponential (in the size of the automata) constant lookahead is both sufficient and in general necessary~\cite{KleinZimmermann16}.
	\item In delay games whose winning conditions are given by deterministic Muller automata, doubly-exponential (in the size of the automata) constant lookahead is sufficient. This follows from the transformation of deterministic Muller automata into deterministic parity automata of exponential size (see, e.g.,~\cite{GraedelThomasWilke02}). However, the best lower bound is the exponential one for parity automata, which are also Muller automata.

 	\item In delay games whose winning conditions are given by LTL formulas, triply-exponential (in the size of the formula) constant lookahead is both sufficient and in general necessary~\cite{KleinZimmermann16a}.

\end{itemize}

Using these facts, we obtain the following complexity results for finite-state strategies: Figure~\ref{fig-naive} shows the upper bounds on the number of states of delay-oblivious finite-state strategies for delay games and on the number of states of delay-aware finite-state strategies for block games and upper bounds on the complexity of determining such strategies. In all three cases, the former strategies are at least exponentially larger and at least exponentially harder to compute. This illustrates the advantage of decoupling tracking the history from managing the lookahead. 

\begin{figure}[h]
\centering
\begin{tabular}{llll}
 &  parity  &  Muller & LTL \\
\toprule
\textbf{delay-oblivious}  & doubly-exp. & 
quadruply-exp.&
quadruply-exp. 
\\

\midrule
\textbf{delay-aware}   & exp.  & 
doubly-exp.&
triply-exp. 

\end{tabular}	
\caption{Memory size for delay-oblivious strategies (for delay games) and delay-aware finite-state strategies (for block games), measured in the size of the representation of the winning condition. For the sake of readability, we only present the orders of magnitude, but not exact values.}
\label{fig-naive}
\end{figure}


Finally, let us compare our approach to that of Salzmann. Fix a delay game~$\delaygame{L(\aut)}$ and assume Player~$I$ has picked $\alpha(0) \cdots \alpha(i)$ while Player~$O$ has picked $\beta(0) \cdots \beta(i')$ with $i'<i$. His strategies are similar to our delay-aware ones for block games. The main technical difference is that his strategies have access to the state reached by $\aut$ when processing ${\alpha(0) \cdots \alpha(i')  \choose \beta(0) \cdots \beta(i')}$. Thus, his strategies explicitly depend on the specification automaton~$\aut$ while ours are independent of it. In general, his strategies are therefore smaller than ours, as our transducers have to simulate~$\aut$ if they need access to the current state. On the other hand, our aggregation-based framework is more general and readily applicable to quantitative winning conditions as well, while he only presents results for selected qualitative conditions like parity, weak parity, and Muller.

%% file: content/implementingstrats.tex
In the previous section, we have shown how to compute delay-aware finite-state strategies for block games via a reduction to Gale-Stewart games. The transducers implementing these strategies process blocks of letters, i.e., the domains of the transition function and of the output function are (roughly) of size~$\size{\SigmaI}^d$ and $\size{\SigmaI}^{2d}$, where $d$ is the block size of the block game. It is known that even for very simple winning conditions, an exponential $d$ is necessary (measured in the size of the automaton~$\aut$ recognizing the winning condition). In this case, the representation of these transducers is at least of doubly-exponential size in $\size{\aut}$, independently of the number of states of the transducer. 

In this section, we propose a succinct notion of transducers implementing delay-aware strategies which can be significantly smaller, e.g., of constant size for the winning condition~$\Leq$ introduced in Section~\ref{sec-fsindg}, which requires Player~$O$ to copy the moves of Player~$I$. Here, the size of the domains of the transition function and of the output function grows exponentially with the block size~$d$, although the transition function and the output function of a transducer implementing a winning strategy are trivial. 

Intuitively, to obtain succinct transducers implementing strategies in block games, we implement the transition function and the output function by transducers. As already alluded to, we present examples (see Examples~\ref{ex-simple}~and~\ref{ex-jpairs-intro}) in which this representation is much smaller than the explicit representation. Furthermore, we give an upper bound on the size of such succinct transducers which is asymptotically equal to the true representation size of explicit transducers in Subsection~\ref{subsec_upperbounds}. Thus, succinct transducers are never larger than explicit ones. However, we also present an example where they cannot be smaller than explicit ones in Subsection~\ref{subsec_lowerbounds}.
Finally, we discuss the relation between block sizes and sizes of succinct transducers in Subsection~\ref{subsec_tradeoff}.

\subsection{Succinct Transducers}
\label{subsec_succincttransducers}

Formally, for a block game~$\blockgame{d}{L}$ with $L \subseteq (\SigmaI \times \SigmaO)^\omega$, we implement a finite-state strategy for Player~$O$ in $\blockgame{d}{L}$ by a transducer~$\strataut = (Q, \SigmaI, q_\initmark, \stratautDelta, \SigmaO, \stratautLambda)$, where $Q$, $\SigmaI$, $q_\initmark$, and $\SigmaO$ are defined as before in Subsection~\ref{subsec-delayaware}.
However, the transition function and the output function are now succinctly represented by transducers $\stratautDelta$ and $\stratautLambda$ which we define below, respectively.
From now on, we refer to this type of transducer as \emph{succinct transducer} and to the type introduced in Subsection~\ref{subsec-delayaware} as explicit transducer.
Regarding succinct transducers, we speak of \myquot{master states} to refer to $Q$, and we speak of \myquot{transition slave} and \myquot{output slave} to refer to $\stratautDelta$ and $\stratautLambda$, respectively.

The transition slave is a tuple $\stratautDelta = (Q_\Delta, \SigmaI, q_\initmark^\Delta, \delta, Q, \lambda)$, where $Q_\Delta$ is a finite  set of states, $q_\initmark^\Delta \colon Q \rightarrow Q_\Delta$ is a function returning an initial state, $\delta \colon Q_\Delta \times \SigmaI \rightarrow Q_\Delta$ is the transition function, and $\lambda \colon Q_\Delta \rightarrow Q$ is the output function. We say that the transition slave computes the function $\fDelta \colon Q \times \SigmaI^* \rightarrow Q$ defined by $\fDelta(q,x) = \lambda(\delta^*(q_\initmark^\Delta(q),x))$, where $\delta^*(q_\initmark^\Delta(q),x)$ is the state  of $\stratautDelta$ reached by processing $x$ from the state~$q_\initmark^\Delta(q)$ of $\stratautDelta$. 
The size of $\stratautDelta$ is defined as $|\stratautDelta| = |Q_\Delta|$.

The output slave is a tuple $\stratautLambda = (Q_\Lambda,\SigmaI,q_\initmark^\Lambda,E,\SigmaO)$, where $Q_\Lambda$ is a finite  set of states, $\SigmaI$ is the input alphabet, $q_\initmark^\Lambda \colon Q \rightarrow Q_\Lambda$ is a function returning an initial state, and $E \colon Q_\Lambda \times (\SigmaI \cup \set{\$}) \rightarrow \SigmaO^* \times Q_\Lambda$ is the deterministic transition function conveniently treated as a relation. Here, $\$$  is a fresh symbol that is used to separate input blocks.

A finite run $\pi$ of $\stratautLambda$ is a sequence \[\pi = (q_0,a_0,b_0,q_1)(q_1,a_1,b_1,q_2)\cdots(q_{i-1},a_{i-1},b_{i-1},q_i) \in E^+.\] We say that $\pi$ starts in $q_0$, ends in $q_i$, its processed input is $\inp{\pi} = a_0\cdots a_{i-1} \in (\SigmaI\cup\set{\$})^+$, and its produced output is $\outp{\pi} = b_0\cdots b_{i-1} \in \SigmaO^*$.
We say that the output slave computes the function $\fLambda \colon Q \times \SigmaI^+ \times \SigmaI^+ \rightarrow \SigmaO^*$ defined by $\fLambda(q,x_1,x_2) = \outp{\pi}$, where $\pi$ is the unique run that starts in $q_\initmark^\Lambda(q)$ with $\inp{\pi} = x_1\$x_2\$$.
The size of $\stratautLambda$ is defined as $|\stratautLambda| = |Q_\Lambda| + \ell$, where $\ell$ is the length of the longest output in $E$, that is, $\mathrm{max}\{ |v| \mid (p,u,v,q) \in E\}$.

Clearly, if additionally $\fLambda(q, \overline{a_0},\overline{a_1}) \in \SigmaO^d$ for every $q \in Q$ and every $\overline{a_0}, \overline{a_1} \in \SigmaI^d$, then the succinct transducer $\strataut$ implements a strategy~$\strat_\strataut$ as before, namely $\strat_\strataut(\block{a_0} \cdots \block{a_{i}}) = \Lambda(\Delta^*(q_\initmark,\block{a_0} \cdots \block{a_{i-2}}),\block{a_{i-1}},\block{a_{i}})$ for $i \geq 1$.
The size of $\strataut$ is defined as $|\strataut| = |Q| + |\stratautDelta| + |\stratautLambda|$.

We illustrate these definitions with two examples. In the first one, we substantiate our above claim that a succinct transducer of constant size implements a winning strategy for Player~$O$ in $\blockgame{d}{\Leq}$, independently of $d$. This is in sharp contrast to explicit transducers implementing winning strategies, whose transition and output function have exponentially-sized domains in $d$.

\begin{example}\label{ex-simple}
  Consider the winning condition $\Leq$ as introduced in Section~\ref{sec-fsindg}.
  Obviously, Player~$O$ can win the block game~$\blockgame{d}{\Leq}$ by copying the moves of Player~$I$ for every block size $d$.
  A succinct transducer implementing a winning strategy for Player~$O$ can be defined independently of $d$.
  One master state, one state for the transition slave, and one state for the output slave suffice; the output slave just copies the input until the first $\$$ occurs and ignores the remaining input. 
\end{example}

One obvious weakness of the previous example is that Player~$O$ does not need lookahead to win a game with winning condition~$\Leq$. Next, we give an example in which Player~$O$ needs lookahead to win, which is obtained by adapting the exponential lower bound on the necessary lookahead in delay games with safety conditions~\cite{KleinZimmermann16}. In this game, Player~$O$ needs exponential lookahead (in the size of an automaton~$\aut$ recognizing the winning condition). Hence, the transition and output function of an explicit transducer have doubly-exponentially-sized domains in $\size{\aut}$. We show how to construct a succinct transducer of exponential size implementing a winning strategy, an exponential improvement.

\begin{example}\label{ex-jpairs-intro}
 Consider the reachability automaton $\aut_n$, for $n > 1$, over the alphabet $\SigmaI \times \SigmaO = \{1,\dots,n\}^2$ of size $\bigo(n)$, given in Figure~\ref{fig-jpairs-intro}.
 The language of $\aut_n$ contains words of the form~${\alpha \choose \beta}$ where $\alpha(1)\alpha(2)\alpha(3) \cdots$ has two occurrences of $\beta(0)$ with only smaller letters in between (a so-called bad $j$-pair for $j=\beta(0)$). Note that the first letter of $\alpha$ is ignored.
 In words, the first letter of the second component indicates the existence of a bad $j$-pair in the $\alpha$-component (again, without its first letter). It is known that Player~$O$ wins $\blockgame{d}{L(\aut_n)}$ for all $d >2^n/2$, but not for smaller ones~\cite{KleinZimmermann16}.

 \begin{figure}[ht!]
\begin{center}
\begin{tikzpicture}[thick]
\tikzstyle{every state}+=[ minimum size=5mm]
\begin{scope}
\tikzstyle{small}=[scale=0.7,draw,gray];
\tikzstyle{scaled}=[scale=0.8];

\node at (-1.1,0) {$\aut_n$};
\node[state,initial]              (a) {};

\node[state,small, above right =of a] (c) {};
\node[state,small,right=of c]         (d) {};
\node[state,small, below right =of a] (e) {};
\node[state,small,right=of e]         (f) {};

\node[state,below right=of d, accepting] (b) {};

\node[scale=0.6,gray] at ($(c)+(-0.3,0.5)$) {$\mathfrak G_1$};
\node[scale=0.6,gray] at ($(e)+(-0.3,0.5)$) {$\mathfrak G_n$};

\node[gray] at ($(a)!0.5!(b)$) {$\vdots$};

\draw[->,gray]
	(c) edge[bend left=15] node {} (d)
	(c) edge[loop above]   node {} ()
	(d) edge[loop above]   node {} ()
	(d) edge[bend left=15] node {} (c);

\draw[->,gray]
	(e) edge[bend left=15] node {} (f)
	(e) edge[loop above]   node {} ()
	(f) edge[loop above]   node {} ()
	(f) edge[bend left=15] node {} (e);

\draw[->]
	(a) edge node[scaled,near start] 
	  {${\ast \choose n}$} (c)
	(a) edge node[scaled,near start,swap] 
	  {${\ast \choose n}$} (e)
	(d) edge node[scaled,near end] 
	  {${1 \choose \ast}$} (b)
	(f) edge node[scaled,near end,swap] 
	  {${n \choose \ast}$} (b)
	(b) edge[loop right] node[scaled]
	  {${\ast \choose \ast}$} ();

\draw[rounded corners,gray,dashed] ($(c)-(0.5,0.4)$) rectangle ($(d)+(0.5,0.7)$) {};
\draw[rounded corners,gray,dashed] ($(e)-(0.5,0.4)$) rectangle ($(f)+(0.5,0.7)$) {};
	
\end{scope}

\begin{scope}[xshift=7.25cm]
\tikzstyle{scaled}=[scale=0.8];

\node[state]  at (0,0) (0) {};
\node[state]  at (2,0) (1) {};
\node at ($(0)+(-1.25,0.75)$) {$\mathfrak G_j$};

\draw[->]
	(0) edge[bend left=15] node[scaled] 
	{${ j \choose \ast}$} (1)
	(0) edge[loop left]   node[scaled] 
	{${\neq j \choose \ast}$} ()
	(1) edge[loop right]   node[scaled]  
	{${< j\choose \ast}$}()
	(1) edge[bend left=15] node[scaled] 
	{${>j \choose \ast}$}(0);

\draw[rounded corners,dashed] ($(0)-(1.75,1.25)$) rectangle ($(1)+(1.75,1.25)$) {};	  
\end{scope}

\end{tikzpicture}
\end{center}
\caption{
 The automaton $\aut_n$ (left) contains gadgets $\mathfrak G_1,\dots,\mathfrak G_n$ (right).
 Transitions not depicted lead to a sink state, which is not drawn.
 The only accepting state is the rightmost state, which is drawn circled. Here, $\ast$ denotes an arbitrary letter from the respective alphabet.
}
\label{fig-jpairs-intro}
\end{figure}

Now, we show that we can construct a succinct transducer implementing a winning strategy in the block game~$\blockgame{d}{L(\aut_n)}$ of exponential size in $n$ for every $d > 2^n/2$.

To begin with, we note that a block size of $d = 2^n/2 +1$ is sufficient in order for Player~$O$ to win the block game~$\blockgame{d}{L(\aut_n)}$, since every word over $\{1,\dots,n\}$ of length at least $2^n$ contains a bad $j$-pair for some $j$~\cite{KleinZimmermann16}.

We now construct a succinct transducer that implements a winning strategy for Player~$O$ in the block game~$\blockgame{d}{L(\aut_n)}$ for every $d > 2^n/2$.
Clearly, to implement a winning strategy, the output slave of a succinct transducer must identify a bad $j$-pair before it can make its first output.
To achieve this, the following transition structure is used: The automaton collects the seen letters, upon reading a letter, all smaller seen letters are deleted from the collection, if a letter is seen that has already been collected, a bad pair has been found.
More formally, from a state $P \subseteq 2^{\{1,\dots,n\}}$ upon reading the letter~$j$ the state $(P \setminus \{1,\dots,j-1\}) \cup \{j\}$ is reached if $j \notin P$, otherwise this is the second occurrence of $j$ in a bad $j$-pair.
Thus, a $j$-pair can be identified using $\bigo(2^n)$ states.
When a bad $j$-pair has been found, the output slave produces an output block of length $d$ beginning with $j$ (in a single computation step) and ignores the remaining input.

There are no conditions for subsequent output blocks, in this case the output slave simply copies the input letter by letter until the first $\$$ occurs and ignores the remaining input.

Thus, the size of an output slave is $\bigo(2^{n})$; the size of a transition slave is constant since it just distinguishes whether the first output block has already been produced.
All in all, the constructed succinct transducer is of size $\bigo(2^{n})$.
\end{example}

\subsection{Upper Bounds}
\label{subsec_upperbounds}

After these two examples showing how succinct transducers can indeed be smaller than explicit transducers, we prove that they do not have to be larger than explicit ones (when measured in the size of the domains of the transition and output function).

\begin{theorem}
\label{thm-transducertransform}
 Let $\stratO$ be a delay-aware finite-state strategy for a block game~$\blockgame{d}{L}$ with $L \subseteq (\SigmaI \times \SigmaO)^\omega$. 
 \begin{enumerate}
	\item\label{thm-transducertransform-explicit2succinct}
	 If $\stratO$ is implementable by an explicit transducer with $n$ states, then also by a succinct transducer with $\bigo(n \cdot \size{\SigmaI}^{2d})$ states.
	\item\label{thm-transducertransform-succinct2explicit}
	 If $\stratO$ is implementable by a succinct transducer with $n$ master states, then also by an explicit transducer with $n$ states. 
\end{enumerate}
\end{theorem}

\begin{proof}
\ref{thm-transducertransform-explicit2succinct}) Let $\strataut = (Q, \SigmaI, q_\initmark, \delta, \SigmaO, \lambda)$ be the explicit transducer implementing $\stratO$, i.e., $\delta \colon Q \times \SigmaI^d \rightarrow Q$ maps a state and a block in $\SigmaI^d$ to a new state and $\lambda \colon Q \times \SigmaI^d \times \SigmaI^d \rightarrow \SigmaO^d$ maps a state and two blocks in $\SigmaI^d$ to a block in $\SigmaO^d$. The strategy is implemented by a succinct transducer over the same set of master states~$Q$, with the same initial state~$q_\initmark$, and where $\delta$ and $\lambda$ are implemented by slaves~$\stratautDelta$ and $\stratautLambda$ defined below.

The transition slave has states of the form~$(q,w) \in Q \times \SigmaI^{\le d}$ to store an input block, an initialization function mapping a state~$q \in Q$ of $\strataut$ to $(q, \epsilon)$, and a transition function mapping a state~$(q, w)$ and a letter~$a \in \SigmaI$ to $(q, wa)$, if $\size{w} < d$. Otherwise, it is mapped to $(q,w)$. This is sufficient, as $\stratautDelta$ is only used to process words of length~$d$. Finally, the output function of the transition slave is defined such that it maps each state~$(q, w)$ with $\size{w} = d$ to $\delta(q,w)$. All other outputs are irrelevant. Then, it is straightforward to prove that the function computed by $\stratautDelta$ (restricted to inputs from $\SigmaI^d$) is equal to $\delta$.

The construction of the output slave~$\stratautLambda$ is analogous: here, we use states of the form~$(q,w)\in Q \times \SigmaI^{\le 2d}$. Again, the initialization function maps~$q$ to $(q,\epsilon)$ and a transition processing a non-$\$$ input letter appends it to the word stored in the state as long as possible. Furthermore, $\$$'s are ignored and transitions from states in $\SigmaI^{2d}$ can be defined arbitrarily. Transitions processing a $\$$ from a state of the form~$(q,\block{a_0}\block{a_1})$ output $\lambda(q, \block{a_0},\block{a_1})$, while all other transitions have an  empty output. Again, it is straightforward to show that the function computed by $\stratautLambda$ coincides with $\lambda$ on inputs of the form~$(q, \block{a_0}, \block{a_1})$.

Hence, the succinct transducer constructed using $\stratautDelta$ and $\stratautLambda$ computes the same function as $\strataut$ and has indeed $\bigo(n \cdot \size{\SigmaI}^{2d})$ states.

\ref{thm-transducertransform-succinct2explicit}) Now, let $(Q, \SigmaI, q_\initmark, \stratautDelta, \SigmaO, \stratautLambda)$ be a succinct transducer implementing $\stratO$. Then, $\stratO$ is also implemented by the explicit transducer~$(Q, \SigmaI, q_\initmark, \delta, \SigmaO, \lambda)$ where $\delta(q,\block{a})$ is equal to the output of $\stratautDelta$ on $\block{a}$ when initialized with $q$, and where $\lambda(q,\block{a_0},\block{a_1})$ is the output of $\stratautLambda$ on $\block{a_0}\$\block{a_1} \$$ when initialized with $q$.
\end{proof}

Thus, we can obtain a succinct transducer by constructing it starting with an explicit one. This explicit one would typically be obtained by the reduction to Gale-Stewart games presented in Section~\ref{sec-construction}. Next, we show how to turn a finite-state strategy for the Gale-Stewart game into a succinct transducer without the detour via explicit transducers, which yields a smaller transducer.

\begin{theorem}
Let $\aut$, $\mon$, $\autb$, and $d$ be as in Theorem~\ref{thm-main} or as in Theorem~\ref{thm-main2}. 

If Player~$O$ has a finite-state winning strategy for the game~$\game(L(\autb))$ with $n$ states, then she has a finite-state winning strategy for $\blockgame{d}{L(\aut)}$ implemented by a succinct transducer of size~$\bigo(n \cdot \size{\SigmaI}^{d} \cdot d)$. 
\end{theorem}	
	
\begin{proof}
Let $Q_\aut$ be the state space of $\aut$ and let $M$ be the set of memory states of $\mon$. Furthermore, let $\strataut = (Q_\strataut, R, q_\initmark^\strataut, \delta_\strataut, Q_\aut \times M_\bot, \lambda_\strataut)$ be a transducer implementing a winning strategy for Player~$O$ in $\game(L(\autb))$. 

Recall the proof of Item~\ref{thm-main-delayfree2block} of Theorem~\ref{thm-main}: there, we turn a finite-state winning strategy for $\game(L(\autb))$ into a delay-aware finite-state winning strategy for $\blockgame{d}{L(\aut)}$ using a simulation: In $\game(L(\autb))$, Player~$I$ picks equivalence classes from $R$ while Player~$O$ picks pairs containing a state of $\aut$ and a memory state of $\mon$. On the other hand, in $\blockgame{d}{L(\aut)}$, both players pick blocks of letters over their respective alphabet. Now, each block of Player~$I$ induces an equivalence class (see Item~\ref{remark-equivbounds-finitelymanynotinR} of Remark~\ref{remark-equivbounds}). For the other direction, we use completions as guaranteed by Remark~\ref{remark-powersetcharac} to translate moves of Player~$O$ in $\game(L(\autb))$ into moves of her in $\blockgame{d}{L(\aut)}$.

Formally, we define a succinct transducer~$\strataut_s = (Q_\strataut, \SigmaI, q_\initmark^\strataut, \stratautDelta, \SigmaO, \stratautLambda)$ simulating $\strataut$. To this end, we just need to specify the slaves~$\stratautDelta$ and $\stratautLambda$.

Intuitively, $\stratautDelta$ computes the transition summary of its input, which represents an equivalence class of $R$, provided the input is long enough. Formally, we define $\stratautDelta = (Q_\stratautDelta, \SigmaI, q_\initmark^\stratautDelta, \delta_\stratautDelta, Q_\strataut, \lambda_\stratautDelta )$ where
\begin{itemize}
	\item $Q_\stratautDelta =  Q_\strataut \times (\pow{Q_\aut \times M_\bot})^{Q_\aut}$,
	\item $q_\initmark^\stratautDelta(q) = (q, q^*\mapsto \set{(q^*, \bot)})$ for $q \in Q_\strataut$ and $q^* \in Q_\aut$, and
	\item $\delta_\stratautDelta((q,r),a) = (q,r')$ with $r'(q^*) = \delta_P(r(q^*),a)$ for every $q^* \in Q_\aut$, where $\delta_P$ is defined as in Section~\ref{sec-construction} on Page~\pageref{page_deltap}.
\end{itemize} 
Thus, we have $\delta_\stratautDelta^*(q_\initmark^\stratautDelta(q),w) = (q,r_{w})$. Note that $\equivclass{\block{a}}$ is an element of $R$ for every block~$\block{a}$ (due to $\size{\block{a}} = d$ and Item~\ref{remark-equivbounds-finitelymanynotinR} of Remark~\ref{remark-equivbounds}).
Hence, we can define $\lambda_\stratautDelta(q,r) = \delta_\strataut(q,\equivclass{w} )$ for some $w$ such that $r_w = r$, if such a $w$ exists. If one does exist, then this definition is independent of the choice of $w$. Otherwise, we define $\lambda_\stratautDelta(q,r)$ arbitrarily. Then, $\stratautDelta$ indeed simulates the transition function $\delta_\strataut$ of $\strataut$.

It remains to define the output slave~$\stratautLambda$. Note that we need to determine a completion of an input block to simulate the strategy implemented by $\strataut$. The \emph{right} completion depends on the block to be completed, not only on its equivalence class. Hence, $\stratautLambda$ needs to store the first block in its input using its state space. For the second block, it suffices to determine its equivalence class, which is implemented as in $\stratautDelta$.

Formally, we define $\stratautLambda = (Q_\stratautLambda, \SigmaI, q_\initmark^\stratautLambda, E_\stratautLambda, \SigmaO )$ with
\begin{itemize}
	\item $Q_\stratautLambda = Q_\strataut \times \SigmaI^{\le d} \times (\pow{Q_\aut \times M_\bot})^{Q_\aut}$,
	\item $q_\initmark^{\stratautLambda}(q) = (q, \epsilon, q^*\mapsto \set{(q^*, \bot)})$ for $q \in Q_\strataut$ and $q^* \in Q_\aut$, and
	\item where $E_\stratautLambda$ is defined such that on inputs of the form~$w \$ w' $ with $\size{w} = d$ the state $(q, w, r_{w'})$ is reached when initializing the run with $q$; on all other inputs an arbitrary state is reached. All these edges have an empty output.
	
The only non-empty output happens on transitions processing a second~$\$ $ from a state of the form~$(q,w,r)$ with $\size{w} = d$ and with $r = r_{w'} \in R$ for some $w'$. If this is the case, let $(q_0,m_0) = \lambda_\strataut(\delta_\strataut(q, \equivclass{w}))$ and $(q_1, m_1) = \lambda_\strataut(\delta_\strataut(\delta_\strataut(q, \equivclass{w}),\equivclass{w'}))$. 
 Then, the output of the transition processing $\$ $ from $(q,w,r)$ is some $(q_0,q_1,m_1)$-completion of $w$. If such a $w'$ does exist, then this definition is independent of the choice of $w'$. 
	
\end{itemize}
Then, $\stratautLambda$ indeed simulates the output function~$\lambda_\strataut$ of $\strataut$.

Altogether, a straightforward induction as in the proof of Item~\ref{thm-main-delayfree2block} of Theorem~\ref{thm-main} shows that $\strataut_s$ indeed implements a winning strategy for the block game.
 \end{proof}
 
\subsection{Lower Bounds}
\label{subsec_lowerbounds}
 
After considering upper bounds in the previous two theorems, we now turn our attention to lower bounds showing that the upper bounds are tight for winning conditions recognized by reachability automata. In this case, an exponential lookahead is sufficient and in general necessary~\cite{KleinZimmermann16}. The following construction is an adaption of the lower bound proof for the lookahead, and again based on bad $j$-pairs. 

\begin{example}\label{ex-jpairs}
 Consider the reachability automaton $\aut_n$, for $n > 1$, over the alphabet $\SigmaI \times \SigmaO = (\{1,\dots,n\} \times \mathbb B^n) \times (\{1,\dots,n\} \times \mathbb B)$ depicted in Figure~\ref{fig-jpairs}.
 The automaton accepts an $\omega$-word \[\bigvec{\alpha \\ \beta_1 \\ \vdots \\ \beta_n \\ \gamma\\ \beta} \in (\SigmaI \times \SigmaO)^\omega\] with $\alpha,\gamma \in \{1,\dots,n\}^\omega$ and $\beta_1,\cdots,\beta_n,\beta \in \mathbb B^\omega$ if, and only if, it has the following form:
 there is an $m$ such that $\alpha(1) \cdots \alpha(m)$ contains a bad $j$-pair for $j=\gamma(0)$, $\alpha(1) \cdots \alpha(m-1)$ contains no bad $j$-pair (which implies $\alpha(m) = j$), and $\beta_j(0) \cdots \beta_j(m) = \beta(0) \cdots \beta(m)$. Intuitively, Player~$O$ has to identify a $j$ such that the $\alpha$-component of the input contains a bad $j$-pair
 and additionally has to copy the $j$th $\beta$-component up to the end of the first $j$-pair. Notice that the first letter of the $\alpha$-component of the input is again ignored when it comes to finding a bad $j$-pair.
 \end{example}

\begin{figure}[ht!]
\begin{center}
\begin{tikzpicture}[thick]
\tikzstyle{every state}+=[minimum size=5mm]
\begin{scope}
\tikzstyle{small}=[scale=0.7,draw,gray];
\tikzstyle{scaled}=[scale=0.6];

\node at (-1.1,0) {$\aut_n$};
\node[state,initial]              (a) {};

\node[state,small, above right =of a] (c) {};
\node[state,small,right=of c]         (d) {};
\node[state,small, below right =of a] (e) {};
\node[state,small,right=of e]         (f) {};

\node[state,below right=of d,accepting] (b) {};

\node[scale=0.6,gray] at ($(c)+(-0.3,0.5)$) {$\mathfrak G_1$};
\node[scale=0.6,gray] at ($(e)+(-0.3,0.5)$) {$\mathfrak G_n$};

\node[gray] at ($(a)!0.5!(b)$) {$\vdots$};

\draw[->,gray]
	(c) edge[bend left=15] node {} (d)
	(c) edge[loop above]   node {} ()
	(d) edge[loop above]   node {} ()
	(d) edge[bend left=15] node {} (c);

\draw[->,gray]
	(e) edge[bend left=15] node {} (f)
	(e) edge[loop above]   node {} ()
	(f) edge[loop above]   node {} ()
	(f) edge[bend left=15] node {} (e);

\draw[->]
	(a) edge node[scaled,near start] 
	  {$\colvec{\ast \\ b_1 \\ \vdots \\ b_n \\ 1 \\ b_1}$} (c)
	(a) edge node[scaled,near start,swap] 
	  {$\colvec{\ast \\ b_1 \\ \vdots \\ b_n \\ n \\ b_n}$} (e)
	(d) edge node[scaled,near end] 
	  {$\colvec{ 1 \\ b_1 \\ \vdots \\ b_n \\ \ast \\ b_1}$} (b)
	(f) edge node[scaled,near end,swap] 
	  {$\colvec{ n \\ b_1 \\ \vdots \\ b_n \\ \ast \\ b_n}$} (b)
	(b) edge[loop right] node[scaled]
	  {$\colvec{\SigmaI \\ \SigmaO}$} ();

\draw[rounded corners,gray,dashed] ($(c)-(0.5,0.4)$) rectangle ($(d)+(0.5,0.7)$) {};
\draw[rounded corners,gray,dashed] ($(e)-(0.5,0.4)$) rectangle ($(f)+(0.5,0.7)$) {};
	
\end{scope}

\begin{scope}[xshift=7.5cm]
\tikzstyle{scaled}=[scale=0.6];

\node[state]  at (0,0) (0) {};
\node[state]  at (2,0) (1) {};
\node at ($(0)+(-1.25,1.75)$) {$\mathfrak G_j$};

\draw[->]
	(0) edge[bend left=15] node[scaled] 
	{$\colvec{ j \\ b_1 \\ \vdots \\ b_n \\ \ast \\ b_j}$} (1)
	(0) edge[loop left]   node[scaled] 
	{$\colvec{\neq j \\ b_1 \\ \vdots \\ b_n \\ \ast \\ b_j}$} ()
	(1) edge[loop right]   node[scaled]  
	{$\colvec{<j\\ b_1 \\ \vdots \\ b_n \\ \ast \\ b_j}$}()
	(1) edge[bend left=15] node[scaled] 
	{$\colvec{>j \\ b_1 \\ \vdots \\ b_n \\ \ast \\ b_j}$}(0);

\draw[rounded corners,dashed] ($(0)-(1.75,2.5)$) rectangle ($(1)+(1.75,2.5)$) {};	  
\end{scope}

\end{tikzpicture}
\end{center}
\caption{
 The automaton $\aut_n$ (left) contains gadgets $\mathfrak G_1,\dots,\mathfrak G_n$ (right).
 Transitions not depicted lead to a sink state, which is not drawn.
 The only accepting state is the rightmost state, which is drawn circled. Here, $\SigmaI$ and $\SigmaO$ denote an arbitrary letter from the respective alphabet.
}
\label{fig-jpairs}
\end{figure}

Using this example, we can prove the following theorem.

\begin{theorem}\label{thm-lowerboundoutput}
 For every $n > 1$, there is a language $L_n$ recognized by a reachability automaton $\aut_n$ with $\bigo(n)$ states such that 
 \begin{itemize}
 	\item Player~$O$ has a finite-state winning strategy in the block game~$\blockgame{d}{L_n}$ for every $d > 2^n / 2$, and   
 	\item every succinct  transducer that implements a winning strategy for Player~$O$ in the block game~$\blockgame{d}{L_n}$ for some $d $ has an output slave with at least $\bigo(2^{n \cdot 2^n})$ states.
 \end{itemize}
\end{theorem}

\begin{proof}
 Consider the reachability automaton $\aut_n$ given in Example~\ref{ex-jpairs}, let $L_n = L(\aut_n)$. 
 To begin with, we argue that Player~$O$ has a finite-state winning strategy in the block game~$\blockgame{d}{L_n}$ for every $d > 2^n / 2$.
 As already mentioned in Example~\ref{ex-jpairs-intro}, every word over $\{1,\dots,n\}$ of length $2^n$ contains a bad $j$-pair for some $j$.
 A block size of at least $2^n/2+1$ allows for a lookahead of at least $2^n+1$ symbols, thus Player~$O$ can correctly identify a bad $j$-pair by remembering the first two input blocks (recall that the first input letter is ignored). This observation suffices to implement a finite-state winning strategy adapting the ideas presented in Example~\ref{ex-jpairs-intro}.
 
 On the other hand, there is a word~$x_n \in \set{1, \ldots, n}^*$ of length~$2^n-1$ that has no bad $j$-pair for every $j \in \set{1, \ldots, n}$~\cite{KleinZimmermann16}. This allows us to prove that Player~$O$ does not win $\blockgame{d}{L_n}$ for any $d \le 2^n / 2$: Player~$I$ can make the first move in the block game using (a prefix, if necessary, of) the word~$1x_n$ in the $\alpha$-component and any bits in the $\beta$-components. Then, Player~$O$ has to pick a first letter~$j^*$ with her first move (all other choices by her are irrelevant to our argument and thusly ignored). In order to win, she has to pick this $j^*$ so that the input has a bad $j^*$-pair. However, since by completing $x_n$ and then playing some $j \neq j^*$ ad infinitum, the outcome does not have a bad $j^*$-pair in its $\alpha$-component, i.e., Player~$I$ wins. For more details, we refer to~\cite{KleinZimmermann16}. 

We use a generalization of this argument to prove the lower bound on the size of the output slave of a finite-state winning strategy for $\blockgame{d}{L_n}$. Hence, let $\strataut = (Q, \SigmaI, q_\initmark, \stratautDelta, \SigmaO, \stratautLambda)$ be a succinct transducer that implements a winning strategy in $\blockgame{d}{L_n}$. As argued above, we can assume $d > 2^n/2$. Towards a contradiction, assume that the output slave~$\stratautLambda = (Q_\Lambda,\SigmaI,q_\initmark^\Lambda,E,\SigmaO)$ has fewer than $ 2^{n \cdot 2^n}$ states.

Recall that $\stratautLambda$ processes words of the form~$x_1\$x_2\$$ where $x_1, x_2 \in \SigmaI^d$ are input blocks. Let $X$ be the set of words of the form
\[\bigvec[\scriptsize]{\alpha(0)\cdots \alpha(2^n-1) \\ \beta_1(0)\cdots \beta_1(2^n-1) \\ \vdots \\ \beta_n(0)\cdots  \beta_n(2^n-1)} \in \SigmaI^{2^n}\] 
with $\alpha(1) \cdots \alpha(2^n-1) = x_n$. We have $\size{X} \ge 2^{n\cdot 2^n}$. 

Hence, there are two words in $X$ that lead $\stratautLambda$ to the same state (when converted into the correct input format for $\stratautLambda$) starting in $q_0 = q_\initmark^\Lambda(q_\initmark)$, which is the initial state used to process the first two blocks. Assume $\stratautLambda$ produces an output during these runs. Then, using arguments as above, one can show that it does not implement a winning strategy, as both words do not contain a bad $j$-pair for any $j$.

Hence, both runs end in the same state and have not yet produced any output. Thus, if both words are extended by the same suffix, $\stratautLambda$ produces the same output for both inputs. Now, let $j^*$ be such that the two words differ in their $\beta_{j^*}$-entry at some position. Consider the extension of the two words by picking $j^*$ in the $\alpha$-component and arbitrary bits in the $\beta$-components, until words of length~$2d$ are obtained. As both inputs only have bad $j$-pairs for $j= j^*$, the automaton has to copy the $\beta_{j^*}$-component. However, it cannot achieve this for both inputs, as it is not able to distinguish the different prefixes. Hence, the automaton does not implement a winning strategy.
\end{proof}

 A note on the size of the automaton $\aut_n$ for $L_n$.
 The number of states is in $\bigo(n)$, but its alphabet is in $\bigo(2^n)$.
 To reduce the size of the alphabet we can consider a variant of $L_n$ defined as follows.
 We call this variant $L'_n$, let $\SigmaI = \SigmaO = \{1,\dots,n,\mathrm t,\mathrm f\}$, we use $\mathrm t$ and $\mathrm f$ in place of $\mathbb B$ to distinguish it from $\{1,\dots,n\}$.
 We are interested in pairs~${\alpha \choose \beta}$ in which the $\alpha$-component is of the form~$\mathrm a_0a_1w_1a_2w_2 \cdots$, where $a_0, a_1, \ldots \in \{1,\dots,n\}$ and $w_1, w_2, \ldots \in \{\mathrm t,\mathrm f\}^n$.
 Meaning, instead of vertical $n$-bit vectors as before, we use horizontal $n$-bit vectors.
 If $\alpha$ is not of this form, then every $\beta$ is allowed in the second component.
 If $\alpha$ is of this form, then ${\alpha \choose \beta } \in L'_n$ if, and only if, $\beta$ is of the form~$b_0 b_1x_1b_2x_2\cdots$, where $b_0,b_1,b_2\ldots \in \{1,\dots,n\}$ and $x_0,x_1,\ldots \in \{\mathrm t,\mathrm f\}^n$ such that if $a_1 \cdots a_i$ is the smallest prefix of $a_1a_2 \cdots$ that contains a bad $j$-pair for $j= b_0$, and additionally the first letter of $x_k$ is the $j$th letter of $w_k$ for $1 \leq k \leq i$.
 
 A reachability automaton $\tilde \aut_n$ for $L'_n$ can be constructed with polynomial size in $n$.
 The idea is to use an automaton similar to the automaton $\aut_n$, and additionally have a ring counter up to $n$ to compare the first bit of $x_k$ with the $j$th bit of $w_k$.
 
 As before, the block game~$\blockgame{d}{L(\tilde \aut_n)}$ can be won by Player~$O$ for any $d$ that allows enough lookahead to identify a bad $j$-pair for some $j$.
 Since every word over $\{1,\dots,n\}$ of length at least $2^n$ contains such a pair, every prefix (in the correct format) of length greater than $2^n\cdot(n+1)$ contains such a pair.
 With the same reasoning as above, a transducer implementing (the output function of) a winning strategy must store every $n$-bit vector until an occurrence of a bad $j$-pair for some $j$ has been witnessed. Thus, the state space of such a transducer is in $\bigo(2^{n\cdot2^n})$.

\subsection{Tradeoff Between Block Size and Memory}
\label{subsec_tradeoff} 
 
Finally, we consider another promising facet of finite-state strategies in delay games: lookahead can be traded for memory and vice versa. Such tradeoffs have previously been presented between lookahead and the semantic quality of winning strategies in games with quantitative winning conditions~\cite{Zimmermann17}, and between memory size and the semantic quality of winning strategies~\cite{WeinertZimmermann17}. With the definition of finite-state strategies, one can add another dimension to the study of tradeoffs in infinite games. 

\begin{theorem}\label{thm-tradeoff}
For every even $k > 0$, there is a language $L_k^R$ recognized by a safety automaton $\aut_k$ such that 
\begin{itemize}
 \item Player~$O$ has a finite-state winning strategy in the block game~$\blockgame{d}{L_k^R}$ for every $d \geq k/2$, 
 \item there exists a succinct transducer $\strataut = (Q, \SigmaI, q_\initmark, \stratautDelta, \SigmaO, \stratautLambda)$ implementing a winning strategy in $\blockgame{d}{L_k^R}$ with $|\stratautDelta| \in \bigo(2^{k-d})$ and $|\stratautLambda| \in \bigo(2^{d})$ for every $d \in \{k/2,\dots,k\}$, and
  \item there exists an explicit transducer $\strataut = (Q, \SigmaI, q_\initmark, \delta, \SigmaO, \lambda)$ implementing a winning strategy in $\blockgame{d}{L_k^R}$  with $|\strataut| \in \bigo(2^{k-d})$ for every $d \in \{k/2,\dots,k\}$.
\end{itemize}
\end{theorem}

\begin{proof}
 We start by describing the language $L_k^R$ over the alphabet $\SigmaI \times \SigmaO = \mathbb B^2$.
 A pair ${\alpha \choose \beta}$ is part of the language if, and only if, $\beta(i) = \alpha((k-1)-i)$ for $0 \leq i \leq k-1$, that is, the first block of length $k$ has to be reversed by Player~$O$.
 
 A safety automaton $\aut_k$ recognizing $L_k^R$ is build as follows.
 Initially, $\aut_k$ stores the first sequence of length $k/2$ in its state space starting from $\left(\varepsilon,\uparrow\right)$ and from a state $\left({a_1 \cdots a_i \choose b_1 \cdots a_i},\uparrow\right)$ upon reading the next letter ${b_{i+1} \choose a_{i+1}}$ it goes to $\left({a_1 \cdots a_{i+1} \choose b_1 \cdots b_{i+1}},\uparrow\right)$ for $0 \leq i \leq k/2-1$.
 Say $\aut_k$ has reached $\left({a_1 \cdots a_{k/2} \choose b_1 \cdots b_{k/2}},\uparrow\right)$, then upon reading the letter ${a_{k/2} \choose b_{k/2}}$ it goes to the state $\left({a_1 \cdots a_{k/2 - 1} \choose b_1 \cdots b_{k/2 -1}},\downarrow\right)$; and to a rejecting sink with any other letter.
 Subsequently, it has to check whether the next sequence of length $k/2 -1$ is equal to ${b_{k/2-1} \cdots b_1 \choose a_{k/2 -1} \cdots a_1}$.
 This can be done checking that in a state $\left({a_1 \cdots a_i \choose b_1 \cdots b_i},\downarrow\right)$ the next read letter is ${b_i \choose a_i}$ and going to $\left({a_1 \cdots a_{i-1} \choose b_1 \cdots b_{i-1}},\downarrow\right)$ for $1 \leq i \leq k/2-1$.
 After reaching $\left(\varepsilon,\downarrow\right)$, any sequence is valid.
 
 It is easy to see that Player~$O$ can win the block game for every $d \geq k/2$.
 Now, for $d \in \{k/2,\dots,k\}$, we show that there exists a succinct transducer $\strataut = (Q, \SigmaI, q_\initmark, \stratautDelta, \SigmaO, \stratautLambda)$ implementing a winning strategy for Player~$O$ in $\blockgame{d}{L_k^R}$  with $|\stratautDelta| \in \bigo(2^{k-d})$ and $|\stratautLambda| \in \mathcal \bigo(2^{d})$.
 Let $x\gamma \in \SigmaI^\omega$ with $x = a_1\cdots a_{k} \in \SigmaI^k$ denote the input sequence that Player~$I$ plays in the block game~$\blockgame{d}{L_k^R}$.
 The first output block that must be produced by Player~$O$ is $a_{k}\cdots{a_{k-d-1}}$.
 This sequence is part of the first lookahead, the output slave $\stratautLambda$ of a succinct transducer $\strataut$ has to store this sequence completely to reverse it, thus $|\stratautLambda| \in \mathcal \bigo(2^d)$.
 The next output block that has to be produced must begin with $a_{k-d}\cdots a_1$.
 This sequence is not part of the next lookahead, it is part of the first input block (the first $k-d$ letters to be precise), the next lookahead is the second and third input block.
 Thus, is must be stored by the transition slave $\stratautDelta$ of $\strataut$ so that this sequence can be passed on to $\stratautLambda$ which has to output it.
 Hence, $\stratautDelta$ has to memorize the first $k-d$ input letters, resulting in $|\stratautDelta| \in \mathcal \bigo(2^{k-d})$.
 
 Regarding explicit transducers, the same reasoning can be applied.
 Thus, in order to implement a winning strategy in the block game~$\blockgame{d}{L_k^R}$, an explicit transducer $\strataut = (Q, \SigmaI, q_\initmark, \delta, \SigmaO, \lambda)$ has to memorize the first $k-d$ input letters, resulting in a state space of size $\bigo(2^{k-d})$.
 Recall, $|Q|$ is defined as the size of $\strataut$, hence $|\strataut| \in \bigo(2^{k-d})$. 
 
 Taking a look at the special cases of $d = k/2$ and $d = k$, the above result yields that an explicit transducer $\strataut$ needs memory of $\bigo(2^{k/2})$ and in the latter case no memory to win the block game $\blockgame{d}{L_k^R}$.
 Generally, for some $d$ between $k/2$ and $k$, an explicit transducer $\strataut$ needs memory of $\bigo(2^{k-d})$ to implement a winning strategy in the block game $\blockgame{d}{L_k^R}$.
 Let us analyze this result; increasing the block size by one halves the number of memory states an explicit transducer needs, thus the tradeoff between the block size and the necessary memory is gradual.
\end{proof}

The example of the block-reversal winning condition $L_k^R$ presented in the proof of Theorem \ref{thm-tradeoff} allows for a tradeoff between the block size and the necessary memory to implement a winning strategy in the block game.
However, the size of an automaton that recognizes $L_k^R$ as well as the lower bound on the block size is exponential in $k$, so the necessary lookahead is only linear in the size of the automaton.
It is an open question whether there is a winning condition recognizable by an automaton of polynomial size with an exponential lower bound on the necessary block size that allows for a tradeoff between block size and memory.

%% file: content/conc.tex
We have presented a very general framework for analyzing delay games. If the automaton recognizing the winning condition satisfies a certain aggregation property, our framework yields upper bounds on the necessary lookahead to win the game, an algorithm for determining the winner (under some additional assumptions on the acceptance condition), and finite-state winning strategies for Player~$O$, if she wins the game at all. These results cover all previous results on the first two aspects (although not necessarily with optimal complexity of determining the winner). 

Thereby, we have lifted another important aspect of the theory of infinite games to the setting with delay. However, many challenging open questions remain, e.g., a systematic study of memory requirements in delay games is now possible. For delay-free games, tight upper and lower bounds on these requirements are known for almost all winning conditions.

Furthermore, in our study we focused on the state complexity of the automata implementing the strategies, i.e., we measure the quality of a strategy in the number of states of a transducer implementing it. However, this is not the true size of such a machine, as this ignores the need to represent the transition function and the output function, which have an exponential domain (in the block size) in the case of delay-aware strategies.
We addressed this issue and have proposed a succinct notion of transducers implementing delay-aware strategies.
Although we have presented examples where our succinct notion allows for a significantly smaller representation of strategies compared to the true size of an explicit representation, generally such a representation cannot be smaller than an explicit one.

Another exciting question concerns the tradeoff between memory and amount of lookahead: can one trade memory for lookahead? In other settings, such tradeoffs exist, e.g., lookahead allows to improve the quality of strategies~\cite{Zimmermann17}. We have presented a game where Player~$O$ can indeed trade lookahead for memory and vice versa. Salzmann has presented further tradeoffs of this kind, e.g., linear lookahead allows exponential reductions in memory size in comparison to delay-free strategies~\cite{Salzmann15}. In current work, we investigate whether these results are inherent to his notion of finite-state strategy, which differs subtly from the one proposed here, or whether they exist in our setting as well.

Finite-state strategies in arena-based games are typically computed by game reductions, which turn a game with a complex winning condition into one in a larger arena with a simpler winning condition. In future work, we plan to lift this approach to delay games. Note that the algorithm for computing finite-state strategies presented here can already be understood as a reduction, as we turn a delay game into a Gale-Stewart game. This removes the delay, but preserves the type of winning condition. However, it is also conceivable that staying in the realm of delay games yields better results, i.e., by keeping the delay while simplifying the winning condition. In future work, we address this question.

\paragraph{Acknowledgements} The authors are very grateful to the anonymous reviewers of this and an earlier version of the paper, which significantly improved the exposition. 

%% file: content/appendix-arenagames.tex
In this short appendix, we give a formal definition of the arena-based games mentioned in Section 3. We begin by giving a quick recap of arena-based games to introduce our notation. 

An arena $\arena = (V, V_I, V_O, E, v_\initmark)$ consists of a finite directed graph~$(V, E)$ without terminal vertices, a partition~$(V_I, V_O)$ into the positions of Player~$I$ and Player~$O$, and an initial vertex~$v_\initmark \in V$. A play is an infinite path through~$\arena$ starting in $v_\initmark$.

A game~$\arenagame = (\arena, \win)$ consists of an arena~$\arena$, say with set~$V$ of vertices, and a winning condition~$\win \subseteq V^\omega$. A play is winning for Player~$O$, if it is in $\win$. 

A strategy for Player~$O$ is a mapping~$\sigma \colon V^*\cdot V_O \rightarrow V$ such that $(v,\sigma(wv)) \in E$ for every $wv \in V^*V_O$. A play~$v_0 v_1 v_2 \cdots$ is consistent with $\sigma$, if $v_{i+1} = \sigma(v_0 \cdots v_i)$ for every $i$ with $v_i \in V_O$. A strategy is winning, of every consistent play is winning for Player~$O$. If Player~$O$ has a winning strategy for $\arenagame$, then we say she wins $\arenagame$.

A finite-state strategy for an arena~$\arena$ with set~$V$ of vertices is again implemented by a transducer~$\strataut = (Q, \SigmaI, q_\initmark, \delta, \SigmaO, \lambda)$ where $Q$, $q_\initmark$, and $\delta$ are as in Subsection~\ref{subsec-finitestate4galestewart}, where $\SigmaI = \SigmaO = V$, and $\lambda \colon Q \times V \rightarrow V$. The strategy implemented by $\strataut$ is defined as $\sigma(wv) = \lambda(\delta^*(q_\initmark,wv),v)$, where $\delta^*(q_\initmark, wv)$ is the state reached by $\strataut$ when processing $wv$ starting in $q_\initmark$. The size of $\strataut$ is defined to be $\size{Q}$. 

A strategy is finite-state if it is implemented  by some finite transducer; it is positional, if it is implemented by some transducer of size one. 

Now, given a Gale-Stewart game~$\game(L(\aut))$ for some automaton~$\aut = (Q, \SigmaI \times \SigmaO, q_\initmark, \delta, \acc)$, we define the arena-based game~$\arenagame_\aut = (\arena_\aut, \win_\aut)$ with $\arena_\aut = (V, V_I, V_O, E, v_\initmark)$ such that:
\begin{itemize}
	\item $V = V_I \cup V_O$ with $V_I = \delta \cup \set{v_\initmark}$ for some fresh initial vertex~$v_\initmark \notin \delta$ and  $V_O = Q \times \SigmaI$.
	\item $E$ is the union of the following sets of edges:
	\begin{itemize}
		\item $\set{(v_\initmark, (q_\initmark, a)) \mid a \in \SigmaI}$ (initial moves of Player~$I$),
		\item $\set{((q, {a \choose b}, q'),(q', a')) \mid (q, {a \choose b}, q') \in V_1, a' \in \SigmaI}$ (regular moves of Player~$I$), and
		\item $\set{ ((q,a), (q, {a \choose b}, q')) \mid (q, a) \in V_0, b\in \SigmaO, q' = \delta(q, {a \choose b}) }$  (moves of Player~$O$).
	\end{itemize} 
	\item $\win_\aut = \set{v_\initmark (q_0,a_0) t_0 (q_1,a_1) t_1 (q_2,a_2) t_2 \cdots \mid t_0 t_1 t_2 \cdots \in \acc} $.
\end{itemize}

The following lemma formalizes a claim from Section~\ref{sec-fsindg}.

\begin{lemma}
\label{lemma-finitestate4galestewart}
Let $\game(L(\aut))$ and $\arenagame_\aut$ be defined as above. Then, Player~$O$ wins $\game(L(\aut))$ if, and only if, she wins $\arenagame_{\aut}$. Furthermore, a finite-state winning strategy with $n$ states for Player~$O$ in $\arenagame_{\aut}$ can be turned into a finite-state winning strategy with $\size{Q} \cdot \size{\SigmaI}\cdot n$ states for Player~$O$ in $\game(L(\aut))$.
\end{lemma}

\begin{proof}
There is a bijection between play prefixes in $\game(L(\aut))$ and in $\arenagame_\aut$. By taking limits, this bijection can be lifted to a bijection  between plays that additionally preserves the winner of plays. Using the former bijection one can easily translate strategies between these games and use the second bijection to prove that this transformation preserves being a winning strategy. Finally, it is also straightforward to implement the transformation from $\arenagame_{\aut}$ to $\game(L(\aut))$ with finite-state strategies: the transducer implementing the strategy for $\game(L(\aut))$ uses a product state space consisting of the states of the given transducer for $\arenagame_{\aut}$ and Player~$O$ vertices from $\arenagame_{\aut}$ to keep track of the last vertex of the play prefix obtained by the first bijection. This information is sufficient to mimic the strategy for $\arenagame_{\aut}$ in $\game(L(\aut))$.
\end{proof}

Now, for some delay game~$\delaygame{L(\aut)}$ for some automaton~$\aut = (Q, \SigmaI \times \SigmaO, q_\initmark, \delta, \acc)$ with constant delay function~$f$ with $f(0) = d > 0$, we define the arena-based game~$\arenagame_{\aut,d} = (\arena_{\aut,d}, \win_{\aut,d})$ with $\arena_{\aut,d} = (V, V_I, V_O, E, v_\initmark)$ such that:
\begin{itemize}
	\item $V = V_I \cup V_O$ with $V_I = \delta \times \SigmaI^{d-1} \cup \set{v_\initmark}$ for some fresh initial vertex~$v_\initmark$ and  $V_O = Q \times \SigmaI^d$.
	\item $E$ is the union of the following sets of edges:
	\begin{itemize}
		\item $\set{(v_\initmark, (q_\initmark, w)) \mid w \in \SigmaI^{d}}$ (initial moves of Player~$I$),
		\item $\set{(((q, {a \choose b}, q'),w),(q', wa')) \mid ((q, {a \choose b}, q'),w) \in V_1, a' \in \SigmaI}$ (regular moves of Player~$I$), and
		\item $\set{ ((q,aw), (q, {a \choose b}, q'),w) \mid (q, aw) \in V_0, b\in \SigmaO, q' = \delta(q, {a \choose b}) }$ (moves of Player~$O$).
	\end{itemize} 
	\item $\win_\aut = \set{v_\initmark (q_0,w_0) (t_0,w_0') (q_1,w_1) (t_1,w_1') (q_2,w_2) (t_2,w_2') \cdots \mid t_0 t_1 t_2 \cdots \in \acc} $.
\end{itemize}

Again, the following lemma  formalizes a claim from Section~\ref{sec-fsindg}.

\begin{lemma}
\label{lemma-finitestate4delay}
Let $\delaygame{L(\aut)}$ and $\arenagame_{\aut,d}$ be defined as above. Then, Player~$O$ wins $\delaygame{L(\aut)}$ if, and only if, she wins $\arenagame_{\aut,d}$. Furthermore, a finite-state winning strategy with $n$ states for Player~$O$ in $\arenagame_{\aut,d}$ can be turned into a finite-state winning strategy with $\size{Q} \cdot \size{\SigmaI}^d\cdot n$ states for Player~$O$ in $\delaygame{L(\aut)}$.
\end{lemma}

\begin{proof}
Similarly to the proof of Lemma~\ref{lemma-finitestate4galestewart}.
\end{proof}